\documentclass[reprint, showpacs, amsmath,amssymb, aps, pra]{revtex4-1}

\usepackage{graphicx}
\usepackage{dcolumn}
\usepackage{bm}

\usepackage{array}								
\usepackage{amssymb}
\usepackage{amsmath}
\usepackage{mathtools}
\usepackage{braket}
\usepackage{amsthm}
\usepackage{dsfont}							
\usepackage{verbatim}						
\usepackage{MnSymbol}						
\usepackage{mathrsfs}							
\usepackage{enumerate}

\newtheorem{thm}{Theorem}

\newtheorem{defn}{Definition}
\newtheorem{lemma}{Lemma}
\newtheorem{corollary}{Corollary}
\theoremstyle{definition}
\newtheorem{example}{Example}

\newenvironment{claim}[1]{\par\noindent\emph{Claim:}\space#1}{}
\newenvironment{claimproof}[1]{\par\noindent\emph{Proof of claim:}\space#1}{\hfill $\blacksquare$}

\newenvironment{corPrime}[1]			
  {\renewcommand{\thecorollary}{\ref{#1}$'$}%
   \addtocounter{corollary}{-1}%
   \begin{corollary}}
  {\end{corollary}}

\newcommand\numberthis{\addtocounter{equation}{1}\tag{\theequation}}
\newcommand{\mket}[1]{\Vert #1 \rrangle}
\DeclareMathOperator{\tr}{Tr}

\newcommand{\rom}[1]{%
  \textup{\uppercase\expandafter{\romannumeral#1}}%
}

\begin{document}

\preprint{APS/123-QED}

\title{Characterization of Gram matrices of multi-mode coherent states}

\author{Ashutosh S Marwah}
 \email{ashutosh.marwah@uwaterloo.ca}
\author{Norbert L\"utkenhaus}%
 \email{nlutkenhaus@uwaterloo.ca}
\affiliation{%
 Institute for Quantum Computing, Department of Physics and Astronomy, University of Waterloo, 200 University Ave W, Waterloo, ON, Canada N2L 3G1
}%
\date{\today}

\begin{abstract}
Quantum communication protocols are typically formulated in terms of abstract qudit states and operations, leaving the question of an experimental realization open. Direct translation of these protocols, say into single photons with some d-dimensional degree of freedom, are typically challenging to realize. Multi-mode coherent states, on the other hand, can be easily generated experimentally. Reformulation of protocols in terms of these states has been a successful strategy for implementation of quantum protocols. Quantum key distribution and the quantum fingerprinting protocol have both followed this route. In this paper, we characterize the Gram matrices of multi-mode coherent states in an attempt to understand the class of communication protocols, which can be implemented using these states. As a side product, we are able to use this characterization to show that the Hadamard exponential of a Euclidean distance matrix is positive semidefinite. We also derive the closure of the Gram matrices, which can be implemented in this way, so that we also characterize those matrices, which can be approximated arbitrarily well using multi-mode coherent states. Using this we show that Gram matrices of mutually unbiased bases cannot be approximated arbitrarily well using multi-mode coherent states.
\end{abstract}

\maketitle

\section{Introduction}

It has been shown that the use of quantum resources offers significant qualitative and quantitative advantages over classical communication for several problems. Quantum cryptographic protocols are a prominent example of the qualitative advantage, while numerous protocols demonstrating the quantitative advantage in terms of communication or information complexity have been developed \cite{Buhrman98, Buhrman01, Yossef04, Gavinsky07}. The theoretical description of these protocols is usually given in terms of d- dimensional quantum systems, or qudits. However manipulation and control of these systems still remains a challenge. Due to this reason the question of experimental implementation of these protocols is also left open. The most successfully implemented quantum protocols are ones which have been reformulated in terms of coherent states, which can be produced using lasers, and linear optics. Widely celebrated quantum key distribution \cite{Grosshans02, Grosshans03, Silberhorn02} and quantum fingerprinting protocols \cite{Arrazola14, Xu15, Guan16} fall in this category, and as always there is a great impetus towards realizing more protocols using these tools \cite{Amiri17, Guan18, Arrazola16}. \\

Consider the quantum fingerprinting protocol as an example. Fingerprinting is a problem in communication complexity. Two parties have to check whether the strings they hold are equal or not in the absence of shared randomness and with as little communication as possible. In 2001, Buhrman et al. gave a quantum protocol for this problem \cite{Buhrman01}, which required exponentially less communication than the optimal classical protocol for this task. It was not until 2015 \cite{Xu15} though, that the protocol was implemented experimentally. This implementation became possible after the protocol had been reformulated such that it required only coherent states and linear optics, in a way which also preserved the communication cost \cite{Arrazola14}. The original protocol mapped the set of n-bit strings, $\lbrace s_i \rbrace_{i=1}^{2^n}$ to qudit states, $\lbrace \ket{\psi_i} \rbrace_{i=1}^{2^n} \subset \mathcal{H}_q$, such that for $i \neq j$, the overlap, $\vert \braket{\psi_i | \psi_j} \vert \leq \delta$ for some $\delta$, and $\dim (\mathcal{H}_q) = O (n)$. Using this mapping, they showed that one could decide if two strings were equal or unequal by communicating only $O(\log n)$ qubits. Arrazola et al. \cite{Arrazola14} showed that one could instead choose multi-mode coherent states which satisfy the requirements on the overlap for the protocol. This naturally leads us to ask under what conditions can the vectors forming a Gram matrix be chosen as multi-mode coherent states, so that a protocol implemented by them may be run using coherent states. \\

This is the question we answer in this paper. We characterize the set of Gram matrices of coherent states and their closure. Not only will such a characterization help us in reformulating quantum protocols in terms of coherent states, but it will also shed light on the fundamental properties of sets of coherent states, which are the most classical states of light \cite{Gerry04, Hillery85, Hudson74}. A Gram matrix of a set of quantum states encodes the information about their orientation relative to each other. For this reason we can very often exchange one set of states with another in a protocol if their Gram matrices are the same, though the measurements also need to be changed correspondingly. Additionally, the set of states attainable by applying physical tranformations on an initial set of states also depends on the Gram matrix of the initial set of states \cite{Chefles04, Chefles00}. Therefore, our work also characterizes all the set of states attainable from a set of multi-mode coherent states. \\

The paper is organised as follows. In section \ref{sec:notation}, we establish the notation for the paper. In Section \ref{sec:Gchar}, we completely characterize the set of Gram matrices, which can be constructed using multi-mode coherent states. This result is stated as Theorem \ref{thm:GcharNew}. We provide a test to check if a matrix belongs in this set. We show that the Hadamard exponential of Euclidean distance matrices can be written as a Gram matrix of multi-mode coherent states, which proves that they are positive semidefinite. Moreover in Section \ref{sec:Gclosure}, we derive the closure of this set to characterize the Gram matrices, which can be approximated arbitrarily well using multi-mode coherent states. This characterization is presented in Theorem \ref{thm:closureG}.

\section{Notation} \label{sec:notation}
In this paper, vectors will be denoted by alphabets, and should be identified by their spaces. For example, $v \in \mathbb{C}^n$ denotes a vector. Fock states will be represented using a Latin alphabet inside a ket, and coherent states using Greek alphabets inside kets. For example $\ket{n}$ is a $n-$photon Fock state, and $\ket{\alpha}$ is a coherent state with complex amplitude $\alpha$. Recall that a coherent state is 
\begin{align*}
	\ket{\alpha} = e^{-\vert \alpha \vert^2/2} \sum_{n=0}^\infty \frac{\alpha^n}{\sqrt{n !}} \ket{n}
\end{align*}
for $\alpha \in \mathbb{C}$. A multi-mode coherent state \cite{Werner06} is simply
\begin{align*}
	\bigotimes_{k=1}^n \ket{\alpha_k}
\end{align*}
for $\lbrace \alpha_k \rbrace_{k=1}^n \in \mathbb{C}$. We will use $\mathscr{C}_n$ to denote the set of n-mode coherent states. 
\begin{align*}
	\mathscr{C}_n := \left\lbrace e^{i \phi} \bigotimes_{k=1}^n \ket{\alpha_k} : \ \forall\  1 \leq k \leq n, \ \alpha_k \in \mathbb{C}, \ \phi \in \mathbb{R} \right\rbrace
\end{align*}
This notation can be simplified, by using the map
\begin{align*}
	& \mket{\ \cdot \ } : \ \mathbb{C}^n \longrightarrow \mathscr{C}_n 
\end{align*}
which takes an amplitude vector, ${\alpha = \left( \alpha_1\ \alpha_2\ \cdots\ \alpha_n \right)^\text{T} \in \mathbb{C}^n}$, and creates a multi-mode coherent state using its components, that is,
\begin{align*}
	& \mket{\alpha } := \bigotimes_{k=1}^n \ket{\alpha_k}
\end{align*}
Hence, $\mathscr{C}_n$ can be written as 
\begin{align*}
	\mathscr{C}_n = \left\lbrace e^{i \phi} \mket{\alpha} : \alpha \in \mathbb{C}^n, \ \phi \in \mathbb{R} \right\rbrace.
\end{align*}
During Section \ref{sec:Gchar}, it will be seen that this notation arises naturally. The following notation will be used to denote the standard inner product on a Hilbert Space, $\mathcal{H}$.
\begin{align*}
	& \langle\ \cdot\ ,\ \cdot\ \rangle\ :\ \mathcal{H} \times \mathcal{H} \longrightarrow \mathbb{C}
\end{align*}
As an example, the inner product between two multi-mode coherent states, $\mket{\alpha}$ and $\mket{\beta}$ will be denoted as $\left\langle \mket{\alpha} , \mket{\beta} \right\rangle$. \\

For a matrix, $P$, the notation $P \geq 0$ will be used to indicate that the matrix is positive semidefinite. We will also define a function which takes an n-tuple of vectors and maps them to their Gram matrix.\\

\begin{defn}
	For a Hilbert space, $\mathcal{H}$, we define $G$ to be the function which takes vectors ${ v_1, v_2, \cdots, v_n  \in \mathcal{H}}$ and maps them to their Gram matrix. 
	\begin{align}
		\nonumber & G : \mathcal{H}^n \rightarrow \text{Pos}(\mathbb{C}^n) \\
		& \left(G\left(v_1, v_2, \cdots, v_n\right)\right)_{ij} := \left\langle v_i, v_j\right\rangle \text{ for } 1 \leq i,j \leq n
	\end{align}
\end{defn}
\noindent Given this definition of $G$, we can introduce the notation, 
\begin{align*}
	G(S^n) = \left\lbrace G(v_1, v_2, \cdots, v_n) : v_1, v_2, \cdots, v_n \in S \right\rbrace.
\end{align*}

\noindent where $S \subset \mathcal{H}$ is a set of vectors. Throughout the paper, we consider $\log$ function to represent the logarithm to the base $e$. Lastly, in Table \ref{tab:Notation}, we list the notation for frequently used sets from linear algebra.

\begin{table}
	\centering
	\caption{Notation for frequently used sets}
	\begin{tabular}{ | m{0.2\columnwidth} | m{0.8\columnwidth} | } 
		\hline
		L$(\mathcal{X}, \mathcal{Y})$ & The set of linear operators from complex Euclidean space, $\mathcal{X}$ to complex Euclidean space, $\mathcal{Y}$.\\ 
		\hline
		L$(\mathcal{X})$ & The set of linear operators from complex Euclidean space, $\mathcal{X}$ to itself. \\ 
		\hline
		U$(\mathcal{X}, \mathcal{Y})$ & The set of isometries from complex Euclidean space, $\mathcal{X}$ to complex Euclidean space, $\mathcal{Y}$.\\ 
		\hline
		Herm$(\mathcal{X})$ & The set of Hermitian operators in L$(\mathcal{X})$.\\ 
		\hline
		Pos$(\mathcal{X})$ & The set of positive semidefinite operators in L$(\mathcal{X})$.\\ 
		\hline
		$[n]$ & The set $\lbrace1,2, \cdots , n \rbrace$\\ 
		\hline
	\end{tabular}
	\label{tab:Notation}
\end{table}

\section{Characterization of Gram matrices of multi-mode coherent states} \label{sec:Gchar}
\begin{sloppypar}
In this section, we will prove a theorem which characterizes Gram matrices of multi-mode coherent states. Namely, we will answer the question: can we write a Gram matrix $P$, as ${P=G(e^{i \phi_1} \mket{\alpha_1}, e^{i \phi_2} \mket{\alpha_2}, \ \cdots \ ,e^{i \phi_n} \mket{\alpha_n} )}$? \\

If indeed the matrix $P$ is a gram matrix of multi-mode coherent states, then for every $i,j $
\begin{align*}
	P_{ij} &= \left\langle e^{i \phi_i} \bigotimes_{k=1}^m \ket{\alpha_{ik}},  e^{i \phi_j} \bigotimes_{k=1}^m \ket{\alpha_{jk}} \right\rangle
\end{align*}
for some given number of modes, $ m \in \mathbb{N}$, a set of amplitudes, $\lbrace \alpha_{jk} : j \in [n], k \in [m]\rbrace$ and a set of real phases, $\lbrace \phi_{i} : i \in [n]\rbrace$. Let us further simplify this:
\begin{align*}
				P_{ij} &= e^{i \left( \phi_j -\phi_i \right)} \prod_{k=1}^m \left\langle \ket{\alpha_{ik}}, \ket{\alpha_{jk}} \right\rangle \\
						&= e^{i \left( \phi_j -\phi_i \right)} \prod_{k=1}^m \exp \left( -\frac{1}{2} \left( \vert \alpha_{ik} \vert^2 + \vert \alpha_{jk} \vert^2 - 2 \alpha_{ik}^\ast \alpha_{jk}\right)\right)	\\
						&= \exp \left( i \left( \phi_j -\phi_i \right) -\frac{1}{2} \sum_{k=1}^m \left( \vert \alpha_{ik} \vert^2 + \vert \alpha_{jk} \vert^2 - 2 \alpha_{ik}^\ast \alpha_{jk}\right)\right).
		\end{align*}
Define, the vector, ${\alpha_i := \left( \alpha_{i1} \ \alpha_{i2} \ \cdots\ \alpha_{im}\right)^\text{T} \in \mathbb{C}^m}$, and the vector of phases, ${\phi := \left( \phi_{1} \ \phi_{2} \ \cdots\ \phi_{n}\right)^\text{T} \in \mathbb{R}^n}$ for every ${i \in [n]}$. This naturally gives rise to the notation mentioned earlier for multi-mode coherent states. The inner product can be written as
		\begin{align}
			P_{ij} & = \exp\left(- \frac{1}{2} \left( \Vert \alpha_i \Vert_2^2 + \Vert \alpha_j \Vert_2^2 - 2 \langle \alpha_i, \alpha_j \rangle\right)+  i \left( \phi_j -\phi_i \right) \right)
			\label{eq:ExpMaster1}
		\end{align}
		or equivalently, 
		\begin{align}
			P_{ij}  & = \exp \left( -\frac{1}{2} \Vert \alpha_i - \alpha_j \Vert_2^2 + i\left( \text{Im}\lbrace \langle \alpha_i, \alpha_j \rangle\rbrace + \left( \phi_j -\phi_i \right)\right) \right). 
			\label{eq:ExpMaster2}
		\end{align}
For a particular branch of the $\log$ functions Eqs. \ref{eq:ExpMaster1} and \ref{eq:ExpMaster2} are equivalent to the existence of $\lbrace N_{ij} \rbrace_{i,j} \subset \mathbb{Z}$ for which 
		\begin{align}
			 \nonumber \log (P_{ij}) - i2\pi N_{ij} = - \frac{1}{2} & \left( \Vert \alpha_i \Vert_2^2 +  \Vert \alpha_j \Vert_2^2 - 2 \langle \alpha_i, \alpha_j \rangle\right) \\
			 & \qquad +  i \left( \phi_j -\phi_i \right)
			 \label{eq:logMaster1}
		\end{align}
		or equivalently,
		\begin{align}
			\nonumber \log (P_{ij}) - i2\pi N_{ij} & = -\frac{1}{2} \Vert \alpha_i - \alpha_j \Vert_2^2 \\
			& + i\left( \text{Im}\lbrace \langle \alpha_i, \alpha_j \rangle\rbrace + \left( \phi_j -\phi_i \right) \right)
			\label{eq:logMaster2}
		\end{align}	
holds for every $i,j \in [n]$. \\

In the Lemma that follows, we characterize matrices $Q$ which satisfy 
\begin{align*}
	Q_{ij} = - \frac{1}{2} \left( \Vert \alpha_i \Vert_2^2 + \Vert \alpha_j \Vert_2^2 - 2 \langle \alpha_i, \alpha_j \rangle\right)+  i \left( \phi_j -\phi_i \right)
\end{align*}
for some set of complex vectors $\lbrace \alpha_i \rbrace_i$, and real phases, $\lbrace \phi_i \rbrace_i$. For the statement of this Lemma, we define the vector, $u \in \mathbb{C}^n$ to be the vector of all ones, i.e.,
\begin{align*}
	u := \left( 1\ 1\ \cdots\ 1\right)^T
\end{align*} 
\begin{lemma}
For a matrix, $Q \in \text{L}(\mathbb{C}^n)$, the following are equivalent:
       \begin{enumerate}[(I)]
       \item There exists $m \in \mathbb{N}$, a set of complex vectors ${\lbrace \alpha_{i} : i \in [n] \rbrace \subset \mathbb{C}^m}$, and a set of real phases ${\lbrace \phi_i : i \in [n] \rbrace \subseteq \mathbb{R}}$ such that for $i,j \in [n]$,
       \begin{equation}
      		Q_{ij} = - \frac{1}{2} \left( \Vert \alpha_i \Vert_2^2 + \Vert \alpha_j \Vert_2^2 - 2 \langle \alpha_i, \alpha_j \rangle\right)+  i \left( \phi_j -\phi_i \right).
      		\label{eq:LemmaQstatement1}
       \end{equation}
       \item $Q \in \text{Herm}(\mathbb{C}^n)$, ${Q_{ii} = 0}$ for all $i \in [n]$ and there exists a vector, $ x \in \mathbb{C}^n $, such that
       \begin{equation}
       			Q + x u^\dagger + u x^\dagger \geq 0.
       \end{equation}
       \item $Q \in \text{Herm}(\mathbb{C}^n)$, ${Q_{ii} = 0}$ for all $i \in [n]$ and for every vector, $ s \in \mathbb{C}^n$, such that $\langle u, s \rangle = 1$, it holds that
       \begin{equation}
				\left( \mathds{1} - u s^\dagger\right) Q \left( \mathds{1} - s u^\dagger\right) \geq 0.
       \end{equation}         
       \item $Q \in \text{Herm}(\mathbb{C}^n)$, ${Q_{ii} = 0}$ for all $i \in [n]$ and there exists a vector, $s \in \mathbb{C}^n$, such that $\langle u, s \rangle = 1$ and
       \begin{equation}
				\left( \mathds{1} - u s^\dagger\right) Q \left( \mathds{1} - s u^\dagger\right) \geq 0.
       \end{equation}   
       \item $Q \in \text{Herm}(\mathbb{C}^n)$, ${Q_{ii} = 0}$ for all $i \in [n]$ and for every vector, $ y \in \mathbb{C}^n$, such that $\langle u, y \rangle = 0$, it holds that
       \begin{equation}
       			y^\dagger Q y \geq 0.
       \end{equation}
       \end{enumerate}
       \label{lemma:Qchar}
\end{lemma}
\begin{proof}
		In the characterization presented here, $Q$ behaves similar to a Euclidean distance matrices \cite{Gower82, Dattorro17}. The equation of the inner products of the vectors in statement \rom{1} of the Lemma is similar to the equation for an element of a distance matrix. The proofs of these statements are almost the same as the ones given by Gower in Ref. \cite{Gower82} for the characterization of Euclidean distance matrices. It should be noted however that $Q$ here is not a distance matrix, since its entries may be complex. We will prove the statements of the theorem in the order
		\begin{align*}
			& (\rom{1}) \Rightarrow (\rom{2}) \Rightarrow (\rom{1}) \\
			& (\rom{2}) \Rightarrow (\rom{3}) \Rightarrow (\rom{4}) \Rightarrow (\rom{2}) \\
			& (\rom{3}) \Rightarrow (\rom{5}) \Rightarrow (\rom{3}).
		\end{align*}
		 We will prove that statement $(\rom{1}) \Rightarrow $ statement $(\rom{2})$. It is clear from Eq. \ref{eq:LemmaQstatement1} that $Q$ is Hermitian and that for every $i \in [n]$, $Q_{ii}=0$. Let $H$ be the Gram matrix of the vectors, $\lbrace \alpha_i \rbrace_{i=1}^n$,  that is, ${H_{ij} := \langle \alpha_i, \alpha_j \rangle}$. Then, we have that $H$ is positive semidefinite, since the set of Gram matrices and positive semidefinite matrices is equivalent. We can write Eq. \ref{eq:LemmaQstatement1} as
		\begin{align}
				Q_{ij} & = i \left( \phi_j -\phi_i \right) -\frac{1}{2} \left( H_{ii} + H_{jj} -2H_{ij} \right).
				\label{eq:MasterEqnInG}
		\end{align}		 
		Without loss of generality let, $H = Q +X$ for some $X$. Then, $H_{ii}= Q_{ii} +X_{ii}= 0+X_{ii} = X_{ii}$. We substitute this into Eq. \ref{eq:MasterEqnInG} to derive a consistency equation for $X$.
		\begin{align*}
				Q_{ij} & = i \left( \phi_j -\phi_i \right) - \frac{1}{2} \left( X_{ii} + X_{jj} -2Q_{ij} - 2X_{ij} \right) \\
				\Rightarrow X_{ij} &= \left( \frac{1}{2}  X_{ii}+i \phi_i\right) + \left(\frac{1}{2}  X_{jj} - i \phi_j\right)
		\end{align*}	
		Define $x \in \mathbb{C}^n$, as $\ x_i  := {X_{ii}}/{2} + i \phi_i \ (X_{ii} = \Vert \alpha_i \Vert_2^2 \in \mathbb{R}$, and $\phi_i \in \mathbb{R})$. Then, for $i,j \in [n]$ we may write
		\begin{align*}
				X_{ij} =  x_i +{x}^\ast_j ,
		\end{align*}
		which is equivalent to 
		\begin{align}
			X =  x u^\dagger +u x^\dagger.
		\end{align}
		Therefore, if statement $(\rom{1})$ is true, then ${H = Q + x u^\dagger +u x^\dagger \geq 0}$. Thus statement $\rom{2}$ is true. \\
		
		For the converse, statement $(\rom{2}) \Rightarrow$ statement $ (\rom{1})$, assume that $Q \in \text{Herm}(\mathbb{C}^n)$, for every $ i \in [n]$ $Q_{ii}=0$ and that there is a vector, ${ x \in \mathbb{C}^n}$, such that ${Q + x u^\dagger +u x^\dagger \geq 0}$. Let $H:= Q + x u^\dagger +u x^\dagger$. Then, there exists a set of vectors, ${\lbrace \alpha_i \rbrace_{i=1}^n}$, such that ${H_{ij}=\langle \alpha_i, \alpha_j \rangle}$, since $H$ is positive semi-definite. We can now show that these amplitude vectors satisfy Eq. \ref{eq:LemmaQstatement1} for an appropriate definition of $\phi$. To see this, observe that
		\begin{align*}
				H_{ij} &= Q_{ij} + x_i +{x}^\ast_j \\
				H_{ii} &=  x_i +{x}^\ast_i.
		\end{align*}
		Now let's evaluate the following expression. 
		\begin{align*}
				\Vert \alpha_i \Vert_2^2+ & \Vert \alpha_j \Vert_2^2 - 2\langle \alpha_i, \alpha_j \rangle \\
				&= H_{ii} + H_{jj} -2H_{ij}\\
				&= x_i + x^\ast_i + x_j + x^\ast_j - 2 Q_{ij} - 2 x_i - 2 x^\ast_j \\
				&= - 2 Q_{ij} - \left(  x_i - x^\ast_i \right) + \left(  x_j - x^\ast_j \right)\\
				&= - 2 Q_{ij} - 2i \text{Im} \lbrace x_i \rbrace +2i \text{Im}\lbrace x_j \rbrace
		\end{align*}
		Define, $\phi \in \mathbb{R}^n$, such that $\phi_i:= \text{Im} \lbrace x_i \rbrace \in \mathbb{R}$. Then, the right hand side of Eq. \ref{eq:LemmaQstatement1} is 
		\begin{align*}
				& -\frac{1}{2} \left( \Vert \alpha_i \Vert_2^2+\Vert \alpha_j \Vert_2^2 - 2\langle \alpha_i, \alpha_j \rangle \right) + i \left( \phi_j - \phi_i \right) \\
				&= \left( Q_{ij} + i\text{Im}\lbrace x_i \rbrace - i \text{Im}\lbrace x_j \rbrace \right) + i \left( \text{Im} \lbrace x_j\rbrace - \text{Im} \lbrace x_i \rbrace \right) \\
				&= Q_{ij}.
		\end{align*}
		Hence, given statement $(\rom{2})$ one can define a set of amplitude vectors, $\lbrace \alpha_i \rbrace_{i=1}^n \subseteq \mathbb{C}^n$, and a vector of phases, $ \phi \in \mathbb{R}^n$ such that Eq. \ref{eq:LemmaQstatement1} is satisfied for all $i, j \in [n]$. Therefore, $(\rom{1}) \Leftrightarrow (\rom{2})$.\\
		
		For all the rest of the statements the facts ${Q \in \text{Herm}(\mathbb{C}^n)}$ and $Q_{ii}=0$ for every $i \in [n]$ are common. So, we don't need to prove them for each statement. We will assume these and prove the rest of the statements. To see that statement $(\rom{2}) \Rightarrow $ statement $(\rom{3})$, choose a vector, $s \in \mathbb{C}^n$, such that $\langle u, s\rangle =1$. We note that for any such choice
		\begin{align*}
				\left( \mathds{1} - u s^\dagger\right) & x u^\dagger \left( \mathds{1} - s u^\dagger\right) \\
				&= \left( \mathds{1} - u s^\dagger\right) \left( x u^\dagger - x (u^\dagger s) u^\dagger\right) \\
				&= 0. \numberthis \label{eq:Proof3}
		\end{align*}
		Using the fact that the expression in Eq. \ref{eq:Proof3} and its conjugate are zero, we can now show, starting with statement $(\rom{2})$, that
		\begin{align*}
				& H := Q + x u^\dagger +u x^\dagger  \geq 0 \\
				& \Rightarrow \left( \mathds{1} - u s^\dagger\right) H \left( \mathds{1} - s u^\dagger\right)  \geq 0 \\
				& \Rightarrow \left( \mathds{1} - u s^\dagger\right) Q \left( \mathds{1} - s u^\dagger\right)  \geq 0.
		\end{align*}
		This proves that statement $(\rom{2}) \Rightarrow $ statement $(\rom{3})$. Statement $(\rom{3}) \Rightarrow $ statement $ (\rom{4})$ trivially. \\ \\
		For statement $(\rom{4}) \Rightarrow $ statement $(\rom{2})$, assume that the vector, $s \in \mathbb{C}^n$ is such that $\langle u, s \rangle=1$, and
		\begin{align*}
				 & \left( \mathds{1} - u s^\dagger\right) Q \left( \mathds{1} - s u^\dagger\right)  \geq 0.
		\end{align*}
		We can write the above as 
		\begin{align*}
				Q + x u^\dagger +u x^\dagger  \geq 0,
		\end{align*}
		for the choice, ${x := \frac{1}{2}\left( s^\dagger Q s\right) u - Q s}$. Hence, we have shown the equivalence of the first four statements. \\ \\
		For statement $(\rom{3})\Rightarrow$ statement $(\rom{5})$, we have that  ${\left( \mathds{1} - u s^\dagger\right) Q \left( \mathds{1} - s u^\dagger\right) \geq 0}$ for every $s \in \mathbb{C}^n$ such that $\langle u, s \rangle =1$ using statement $(\rom{3})$. Then, for every vector, ${y \in \mathbb{C}^n}$, such that, ${\langle u, y\rangle = 0 }$, 
		\begin{align*}
			& y^\dagger \left( \mathds{1} - u s^\dagger\right) Q \left( \mathds{1} - s u^\dagger\right)   y \geq 0 \\
			& \Rightarrow y^\dagger Q y \geq 0.
		\end{align*}
		For statement $(\rom{5}) \Rightarrow $ statement $(\rom{3})$, we assume that for every vector, ${y \in \mathbb{C}^n}$, such that, ${\langle u, y\rangle = 0}$,
		\begin{align*}
			y^\dagger Q y & \geq 0.
		\end{align*}
		If we choose any vector $s\in \mathbb{C}^n$ such that $\langle u, s\rangle=1$, and any vector $v \in \mathbb{C}^n$, then we can construct a vector $y$ with $\langle u, y \rangle=0$ by setting $y= (\mathds{1}- s u^\dagger) v$. Using statement $(\rom{5})$ with this choice of vector y then gives
		\begin{align*}
			& v^\dagger \left( \mathds{1} - u s^\dagger\right) Q  \left( \mathds{1} - s u^\dagger\right) v \geq 0 \\
			& \Rightarrow \left( \mathds{1} - u s^\dagger\right) Q \left( \mathds{1} - s u^\dagger\right) \geq 0,
		\end{align*}
		which proves statement $(\rom{5}) \Rightarrow $ statement $(\rom{3})$.
		
\end{proof}

\noindent We can use this Lemma to check if a matrix is a Gram matrix of multi-mode coherent states. First, we define the Hadamard logarithm, as 
\begin{align*}
	(\log\odot (P))_{ij} := \log P_{ij}.
\end{align*}
Using Eq. \ref{eq:logMaster1} and Lemma \ref{lemma:Qchar}, we have that a matrix $P$ such that $P_{ii}=1$ for every $i \in [n]$ is Gram matrix of coherent states if and only if there exists a matrix of integers $N \in \mathbb{Z}^{n \times n}$ such that ${(\log \odot(P)-2\pi iN) \in \text{Herm}(\mathbb{C}^n)}$, ${2 \pi i N_{ii} = \log 1}$ for every $i \in [n]$ and  
\begin{align*}
	\left( \mathds{1} - \frac{u u^\dagger}{n}\right) (\log \odot(P)-2\pi iN) \left( \mathds{1} - \frac{u u^\dagger}{n}\right) \geq 0.
\end{align*}
However, in this form this condition cannot be checked since there are countably infinite choices for the matrix, $N$. Now, we will show that we only need to check this condition for a finite number of matrices, $N$. First, we show that if a matrix can be written as a Gram matrix of multi-mode coherent states, then one of the coherent states can be chosen as $\mket{0}$.
\begin{lemma}
	If ${P=G(e^{i \phi_1} \mket{\alpha_1}, e^{i \phi_2} \mket{\alpha_2}, \ \cdots \ ,e^{i \phi_n} \mket{\alpha_n} )}$, then ${P=G(e^{i \phi'_1} \mket{0}, e^{i \phi'_2} \mket{\alpha'_2}, \ \cdots \ ,e^{i \phi'_n} \mket{\alpha'_n} )}$, where ${\alpha'_i = \alpha_i-\alpha_1}$, and ${\phi'_i = \phi_i - \text{Im}\lbrace \langle \alpha_i ,\alpha_1 \rangle\rbrace}$ for every ${i \in [n]}$.
	\label{lemma:GAlphaZero}
\end{lemma}
\begin{proof}
	This can be proven by just plugging in the values given above in Eq. \ref{eq:ExpMaster2}. 
\end{proof}
Our goal is to restrict the values the matrix element $N_{ij}$ can take. We accomplish this by using a complex $\log$ function mapping to the branch, $[\beta, \beta+2\pi)$ in the following equation. 
\begin{align}
			\nonumber \log (P_{ij}) & = -\frac{1}{2} \Vert \alpha_i - \alpha_j \Vert_2^2 \\
			& + i\left( \text{Im}\lbrace \langle \alpha_i, \alpha_j \rangle\rbrace + \left( \phi_j -\phi_i \right) + 2\pi N_{ij} \right)
			\label{eq:logMaster2Right}
		\end{align}	
The imaginary part on the left hand side of Eq. \ref{eq:logMaster2Right} is restricted to the interval $[\beta, \beta+2\pi)$. If we can bound the terms in the imaginary part on the right hand side, then we would be able to bound the terms $N_{ij}$ as well. For this purpose, we define the parameter, 
\begin{align}
	\delta := \min_{i,j} |P_{ij}|.
	\label{eq:deltaDefn}
\end{align}
If $\delta=0$, then we know that ${P \notin G (\mathscr{C}_m^n)}$, since the inner product between two coherent states cannot be zero (see Eq. \ref{eq:ExpMaster1}). So, we consider $\delta \neq 0$. If ${P \in G (\mathscr{C}_m^n)}$, then ${P=G(e^{i \phi_1} \mket{\alpha_1}, e^{i \phi_2} \mket{\alpha_2}, \ \cdots \ ,e^{i \phi_n} \mket{\alpha_n} )}$ for some amplitude vectors, $\lbrace \alpha_i \rbrace_{i=1}^n$ with $\alpha_1=0$ and real phases, $\lbrace \phi_i \rbrace_{i=1}^n$ (using Lemma \ref{lemma:GAlphaZero}). Moreover, we can assume that for every $i \in [n]$, ${\phi_i \in [0, 2\pi)}$. Then, for every $i,j \in [n]$ we have that
\begin{align}
	-2\pi \leq \phi_j - \phi_i \leq 2\pi.
	\label{eq:AngleBound}
\end{align}
For a matrix ${P \in G (\mathscr{C}_m^n)}$, we have
\begin{align*}
	\delta &= \min_{i,j} \lbrace \exp \left( -\frac{1}{2} \Vert \alpha_i -\alpha_j \Vert_2^2 \right) \rbrace \\
	& =  \exp \left( -\frac{1}{2} \left( \max_{i,j} \lbrace \Vert \alpha_i -\alpha_j \Vert_2 \rbrace \right)^2 \right).
\end{align*}
This implies that 
\begin{align*}
	\max_{i,j} \lbrace \Vert \alpha_i -\alpha_j \Vert_2 \rbrace  = (-2 \log (\delta))^{1/2}.
\end{align*}
For every $i \in [n]$, and amplitude vector, $\alpha_i$ in this representation of $P$, we have
\begin{align*}
	\Vert \alpha_i \Vert_2 & = \Vert \alpha_i -0 \Vert_2 \\
	 & = \Vert \alpha_i -\alpha_1 \Vert_2 \\
	 & \leq \max_{i,j} \lbrace \Vert \alpha_i -\alpha_j \Vert_2 \rbrace= (-2 \log (\delta))^{1/2}.
\end{align*}
Further, we have that for every $i,j \in [n]$
\begin{align*}
	\vert \langle \alpha_i , \alpha_j \rangle \vert  & \leq \Vert \alpha_i \Vert_2 \Vert \alpha_j \Vert_2 \\
	&\leq -2 \log (\delta) = \vert 2 \log (\delta) \vert. 
\end{align*}
We will use this to bound the Im$\lbrace \alpha_i ,\alpha_j\rbrace$, using 
\begin{align}
	\vert \text{Im} \lbrace \langle \alpha_i ,\alpha_j \rangle \rbrace \vert \leq \vert \langle \alpha_i ,\alpha_j\rangle \vert \leq \vert 2 \log (\delta) \vert
	\label{eq:IPBound}
\end{align}
Now, observe that in Eq. \ref{eq:logMaster2Right} if we consider the $\log$ function with $[\beta, \beta+2\pi)$ branch, for every $i,j \in [n]$ we have 
\begin{align*}
	& \beta \leq 2\pi N_{ij} + \text{Im}\lbrace \langle \alpha_i , \alpha_j \rangle \rbrace + (\phi_j-\phi_i) \\
	\Rightarrow\ & \beta - \text{Im}\lbrace \langle \alpha_i , \alpha_j \rangle \rbrace - (\phi_j-\phi_i) \leq 2\pi N_{ij} \\
	\Rightarrow\ & \beta - \vert 2 \log (\delta) \vert - 2\pi \leq 2\pi N_{ij},
\end{align*}
and, 
\begin{align*}
	& 2\pi N_{ij} + \text{Im}\lbrace \langle \alpha_i , \alpha_j \rangle \rbrace + (\phi_j-\phi_i) < \beta +2\pi \\
	\Rightarrow\ & 2\pi N_{ij} < \beta +2\pi - \text{Im}\lbrace \langle \alpha_i , \alpha_j \rangle \rbrace - (\phi_j-\phi_i) \\
	\Rightarrow\ & 2\pi N_{ij} < \beta +4\pi + \vert 2 \log (\delta) \vert,
\end{align*}
for which we have used Eqs. \ref{eq:AngleBound} and \ref{eq:IPBound}. Thus for every $i,j \in [n]$ we have the following bound: 
\begin{align}
	N_{ij} \in \mathbb{Z} \cap \left[ \frac{\beta}{2\pi} - \frac{1}{\pi}\vert \log (\delta) \vert - 1, \frac{\beta}{2\pi} + \frac{1}{\pi}\vert \log (\delta) \vert +2 \right).
	\label{eq:NBounds}
\end{align}
To summarize, we have proven that if ${P \in G (\mathscr{C}_m^n)}$, then there exists an integer matrix $N \in \mathbb{Z}^{n \times n}$ with elements in the range given by Eq. \ref{eq:NBounds} and a set of vectors, $\lbrace \alpha_i \rbrace_i$ (where $\alpha_1 =0$ specifically) and phases, $\lbrace \phi_i \rbrace_i$, such that for every $i,j$
\begin{align*}
\log (P_{ij}) - i2\pi N_{ij} = - \frac{1}{2} & \left( \Vert \alpha_i \Vert_2^2 +  \Vert \alpha_j \Vert_2^2 - 2 \langle \alpha_i, \alpha_j \rangle\right) \\
			 & \qquad +  i \left( \phi_j -\phi_i \right).
\end{align*}
Using the characterization of matrix equations of this form given in Lemma \ref{lemma:Qchar}, we have that this is equivalent to ${(\log\odot(P)-2\pi iN)}$ being a Hermitian matrix, ${2\pi iN_{ii}=\log(1)}$ for every $i \in [n]$ and 
\begin{align*}
	\left( \mathds{1} - \frac{u u^\dagger}{n}\right) (\log\odot(P)-2\pi iN) \left( \mathds{1} - \frac{u u^\dagger}{n}\right) \geq 0.
\end{align*}
Thus, we have proven the following Theorem, which characterizes the Gram matrices of multi-mode coherent states. Here we consider the $\log$ function with the branch $[\beta, \beta+2\pi)$.
\begin{thm}
For a matrix, $P \in \text{Herm}(\mathbb{C}^n)$ such that ${P_{ii} = 1}$ for all $i \in [n]$, $P \in G(\mathscr{C}_m^n)$ for some $m \in \mathbb{N}$ if and only if there exists an integer matrix $N \in\mathbb{Z}^{n \times n}$, such that ${(\log \odot (P) - 2 \pi i N) \in  \text{Herm}(\mathbb{C}^n)}$, $2 \pi i N_{ii} = \log(1)  $ for every $i \in [n]$, and 
	\begin{align}
		\left( \mathds{1} - \frac{u u^\dagger}{n}\right) (\log\odot(P)-2\pi iN) \left( \mathds{1} - \frac{u u^\dagger}{n}\right) \geq 0.
		\label{eq:GcharNew2}
	\end{align}
	Further, we can restrict the range of the elements of $N$ to the following:
		\begin{align}
		N_{ij} \in \mathbb{Z} \cap \left[ \frac{\beta}{2\pi} - \frac{1}{\pi}\vert \log (\delta) \vert - 1, \frac{\beta}{2\pi} + \frac{1}{\pi}\vert \log (\delta) \vert +2 \right)
		\label{eq:GcharNew1}
	\end{align}
	for every $i,j \in [n]$, where $\delta := \min_{i,j} |P_{ij}| >0$.
	\label{thm:GcharNew}
\end{thm}
Note that Eq. \ref{eq:GcharNew2} can be replaced with any of the statements from Lemma \ref{lemma:Qchar}. We have used $s= u/n$ in the statements \rom{3} and \rom{4} of Lemma \ref{lemma:Qchar} for simplicity. Moreover, for every matrix $P \in \text{Herm}(\mathbb{C}^n)$ we only need to check Eq. \ref{eq:GcharNew2} for finitely many matrices, $N$. However, we need to check the conditions for $\exp(O(n^2))$ number of matrices $N$, where $n$ is the size of the matrix $P$. It might be possible to use semidefinite programming and rounding methods to create a more efficient algorithm for this task but we do not pursue this lead here.\\

In the following corollary, we show that we can restrict the number of modes of the coherent states forming a Gram matrix to the size of the matrix.
\begin{corollary}
	For every $ m \in \mathbb{N}$, $G\left( \mathscr{C}_{n+m}^{n}\right) = G\left( \mathscr{C}_{n}^{n}\right)$. \\
	That is, no more than $n$-modes are required to represent a Gram matrix of $n$-vectors.
	\label{corr:GOnlyN}
\end{corollary}
\begin{proof}
	A matrix $P \in \text{Herm}(\mathbb{C}^n)$ with $P_{ii}=1$ for every $i \in [n]$ can be written as a Gram matrix of coherent states if and only if there exist complex amplitude vectors, $\lbrace \alpha_i \rbrace_{i=1}^n \subset \mathbb{C}^m$ (where $m$ is the number of modes) and a set of real phases $\lbrace \phi_i \rbrace_{i=1}^n$ such that  
	\begin{align*}
			P_{ij} & = \exp\left(- \frac{1}{2} \left( \Vert \alpha_i \Vert_2^2 + \Vert \alpha_j \Vert_2^2 - 2 \langle \alpha_i, \alpha_j \rangle\right)+  i \left( \phi_j -\phi_i \right) \right).
	\end{align*}
	It is clear from the above equation that the collection of vectors, ${\lbrace \alpha_i \rbrace_{i=1}^n}$, is isometrically invariant, i.e.,  if $\lbrace \alpha_i \rbrace_{i=1}^n$ satisfy this equation, then for a ${ U \in U(\mathbb{C}^m, \mathbb{C}^p)}$, the vectors ${\lbrace U \alpha_i \rbrace_{i=1}^n}$ also satisfy this equation. Therefore, one can simply constrain the vectors, $\lbrace \alpha_i \rbrace_{i=1}^n$, to be in an $n$-dimensional space (since there are only $n$ vectors), i.e., we can choose the number of modes $m=n$. \\
\end{proof}

Corollary \ref{corr:GOnlyN} tells us that one cannot do better by simply increasing the number of modes. With this result in hand, we can see that the set of Gram matrices, which can be constructed using coherent states, is just $G(\mathscr{C}_n^n)$. Therefore, in Section \ref{sec:Gclosure}, where we study the closure of the set of Gram matrices of multi-mode coherent states, we will just consider the set $G(\mathscr{C}_n^n)$.\\

We will demonstrate our characterization result by using it to show that $n$ multi-mode coherent states, $\lbrace e^{i \phi_i}\mket{\alpha_i} \rbrace_{i=1}^n$ can be chosen such that for $ i \neq j$ the inner product $\langle e^{i \phi_i} \mket{\alpha_i}, e^{i \phi_j}\mket{\alpha_j} \rangle = r$ for some $r \in (0,1]$. That is, we can choose acute equiangular multi-mode coherent states.

\begin{example}
	\emph{Acute equiangular coherent states} \\
	To show that coherent states with the aforementioned property exist, we show that the Gram matrix, ${P \in \text{Herm}(\mathbb{C}^n)}$ such that $P_{ii} =1$ for every $i \in [n]$, and $P_{ij} = r$ (where $r \in (0,1]$) for every $i \neq j \in [n]$, can be constructed using coherent states. For this we choose the $\log$ function with $[-\pi, \pi)$ branch. Then, $\log\odot P$ is Hermitian. Moreover, $\log P_{ii} =0$ for every $ i \in [n]$, and $\log P_{ij} = \log (r)$ for every $i \neq j \in [n]$. That is, 
	\begin{align*}
		\log \odot P = \log (r) u u^\dagger - \log (r) \mathds{1}.
	\end{align*}
	Then, we have that 
	\begin{align*}
		& \left( \mathds{1} - \frac{u u^\dagger}{n}\right)(\log \odot P ) \left( \mathds{1} - \frac{u u^\dagger}{n}\right) \\
		& \quad  = \log \left(\frac{1}{r}\right)\left( \mathds{1} - \frac{u u^\dagger}{n}\right) \geq 0.
	\end{align*}
	Thus, we can choose $n$ coherent states, such that the inner products between any two of these states is equal to some $r \in (0,1]$. 
	\label{ex:Ex1}
\end{example}

We can also use our characterization to prove that the Hadamard exponential of a Euclidean distance matrices is positive semidefinite. We will formulate this result as Corollary \ref{corr:EDM}. A matrix, ${D \in \text{L}(\mathbb{R}^n)}$ is a Euclidean distance matrix if there exist vectors, ${\lbrace \alpha_i \rbrace_{i=1}^n \subset \mathbb{R}^n}$, such that for ${i,j \in [n]}$
\begin{align*}
	D_{ij} = -\frac{1}{2} \Vert \alpha_i - \alpha_j \Vert_2^2.
\end{align*}
In Ref. \cite{Gower82} it has been proven that a matrix is a Euclidean distance matrix if and only if  
\begin{align}
	\left( \mathds{1} - \frac{uu^\dagger}{n}\right) D \left( \mathds{1} - \frac{uu^\dagger}{n}\right) \geq 0  .
	\label{eq:DistMatrixTest}   
\end{align}
Lastly, we define the Hadamard exponential of a matrix, $X$ as the component wise exponential of a matrix. That is,
\begin{align*}
	(\exp\odot (X))_{ij} :=\exp (X_{ij}).
\end{align*}
\begin{corollary}
	The Hadamard exponential of a Euclidean distance matrix, $D$ is positive semidefinite. That is
	\begin{align}
			\exp \odot (D) \geq 0.
	\end{align}
	\label{corr:EDM}
\end{corollary}
\begin{proof}
	Let $P := \exp \odot (D)$. Since, $D$ is a Euclidean distance matrix, it is a real matrix and $P_{ij}>0$ for every $i,j$. If we choose the $\log$ function to be the logarithm with the branch $[-\pi, \pi)$, then $\log \odot (P) = D$. Using this fact, we have
	\begin{align*}
		& \left( \mathds{1} - \frac{uu^\dagger}{n} \right) D \left( \mathds{1} - \frac{uu^\dagger}{n} \right) \geq 0 \\
		\Rightarrow & \left( \mathds{1} - \frac{uu^\dagger}{n} \right) (\log \odot (P)) \left( \mathds{1} - \frac{uu^\dagger}{n} \right) \geq 0 \\
		\Rightarrow & \ P \in G(\mathscr{C}_n^n) \\
		\Rightarrow & \ P \geq 0 \\
		\Rightarrow & \ \exp \odot (D) \geq 0,
	\end{align*}
	where we have used Theorem \ref{thm:GcharNew} in the third step and the fact that Gram matrices are positive semidefinite in the fourth step.
\end{proof}

One can, in principle use Theorem \ref{thm:GcharNew} to not only check if a given square matrix, $P \in \text{L}(\mathbb{C}^n)$ (such that $P_{ii}=1$ for every $i \in [n]$) can be represented as a Gram matrix of coherent states or not, but also to find a set of coherent states, $\lbrace e^{i \phi_i}\mket{\alpha_i} \rbrace_i$ such that ${P=G(e^{i \phi_1}\mket{\alpha_1},\ \cdots\ , e^{i \phi_n}\mket{\alpha_n})}$. One can use the following algorithm for example. Choose $\log$ to be the logarithm with the branch $[-\pi, \pi)$. For the given matrix $P$, calculate $\delta := \min_{i,j} \vert P_{ij} \vert $. Let
\begin{align*}
	\mathcal{F}:=& \lbrace N\in \mathbb{Z}^{n \times n} : \\
	& N_{ij} \in \mathbb{Z} \cap \left[ -\frac{3}{2} - \frac{1}{\pi}\vert \log (\delta) \vert , \frac{3}{2} + \frac{1}{\pi}\vert \log (\delta) \vert \right)\\
	& \text{ for every  } i \neq j \in [n],\ 2 \pi i N_{ii} = \log(1) \\
	& \text{ for every  } i \in [n] \text{ and } \\
	& (\log \odot (P) -2 \pi i N) \in \text{Herm}(\mathbb{C}^n) \rbrace.
\end{align*}
be the set of possible integer matrices. Note that this set is finite. For every $N \in \mathcal{F}$ do the following:
\begin{enumerate}
	\item Let ${X:= ( \mathds{1} - u e_1^\dagger) \left( \log\odot (P) -2\pi i N \right) ( \mathds{1} - e_1 u^\dagger)}$.
	\item Check if $X$ is positive semidefinite or not.
		\begin{enumerate}
			\item If it is positive semidefinite, then Theorem \ref{thm:GcharNew} (we are using an alternate but equivalent version of the condition in Eq. \ref{eq:GcharNew2} for convenience) guarantees that you can write your matrix as a Gram matrix of multi-mode coherent states. Moreover, these coherent states can be found by using the columns of $X^{1/2}$ (this can be seen through the proof of Lemma \ref{lemma:Qchar}). If
			\begin{align*}
				X^{1/2} = \left( \alpha_1\ \alpha_2\ \cdots\ \alpha_n\right)
			\end{align*}
			where ${\alpha_i \in \mathbb{C}^n}$ for each ${i \in [n]}$, then we have ${P = G(e^{i \phi_1} \mket{\alpha_1}, e^{i \phi_2} \mket{\alpha_2},\ \cdots\ ,e^{i \phi_n} \mket{\alpha_n})}$ for ${\phi_i := -\text{Im}\lbrace \log P_{i1} \rbrace}$ for each ${i \in [n]}$.
			\item  If it is not positive semidefinite, move to the next integer matrix in $\mathcal{F}$. 
		\end{enumerate}
\end{enumerate}
Finally, if $X$ is not positive semidefinite for any of the integer matrices $N \in \mathcal{F}$, then $P \notin G(\mathscr{C}_{n}^n)$.\\

\emph{Remark: } If we replace $e_1$ with $u/n$ in the above algorithm, then $X$ and $X^{1/2}$ will have rank at most $(n-1)$, which means that one could in fact choose vectors in $\mathbb{C}^{n-1}$ to satisfy Eq. \ref{eq:ExpMaster1}. Thus, we can choose $(n-1)-$mode coherent states to construct any matrix in $G(\mathscr{C}_n^n)$ or, ${G(\mathscr{C}_n^n) = G(\mathscr{C}_{n-1}^n)}$. One can also see this through Lemma \ref{lemma:GAlphaZero}, which allows us to reduce the rank of the amplitude vectors by at least 1. We state this result as Corollary \ref{corr:GOnlyN2}, a slightly strengthened form of Corollary \ref{corr:GOnlyN}. 

\begin{corPrime}{corr:GOnlyN}
	For $n,m \in \mathbb{N}$ and $n \geq 2$, ${G\left( \mathscr{C}_{n-1+m}^{n}\right) = G\left( \mathscr{C}_{n-1}^{n}\right)}$.
	\label{corr:GOnlyN2}
\end{corPrime}

However, we will continue to use $G\left( \mathscr{C}_{n}^{n}\right)$ to represent the Gram matrices of $n$ multi-mode coherent states for notational convenience. \\

Since, it seems hard to decide if a Gram matrix $P$ lies in $G(\mathscr{C}_n^n)$ or not, we provide two Corollaries, which would be helpful in deciding this question in certain cases. Both of these rely on the fact that it is easy to check if a matrix is a Euclidean distance matrix (EDM) (see Eq. \ref{eq:DistMatrixTest}). 
\begin{corollary}
	If $P \in G(\mathscr{C}_n^n)$, then the matrix ${R:=[\log(\vert P_{ij} \vert)]_{ij}}$ is a Euclidean distance matrix ($\log$ function considered here is the real logarithm function). \label{corr:NewCorr1}
\end{corollary}
\begin{proof}
	If $P \in G(\mathscr{C}_n^n)$, then for every $i,j \in [n]$
	\begin{align}
		\log(\vert P_{ij} \vert) = -\frac{1}{2} \Vert \alpha_i - \alpha_j \Vert_2^2
		\label{eq:Corr1Eq}
	\end{align}
		for some complex vectors, $\lbrace \alpha_k \rbrace_{k=1}^n \subset \mathbb{C}^n$. Define, a $2n$- dimensional vector $\beta_i := (\text{Re}\lbrace \alpha_i \rbrace, \text{Im}\lbrace \alpha_i \rbrace)$ for every ${i \in [n]}$. These vectors would also satisy Eq. \ref{eq:Corr1Eq}. Further these can be rotated into a n-dimensional space using an orthogonal matrix, which would leave the distance between these vectors invariant. Hence, the matrix $R$ would be an EDM. 
\end{proof}
One may wonder if the converse also holds in Corollary \ref{corr:NewCorr1}, i.e., given that for a Gram matrix $P$, the matrix ${R:=[\log(\vert P_{ij} \vert)]_{ij}}$ is a Euclidean distance matrix, does this imply $P \in G(\mathscr{C}_n^n)$? This is seen not to be the case. A counterexample is the Gram matrix of three equiangular vectors in two dimensions,
\begin{align*}
	P = 
	\begin{pmatrix}	
		1 & -0.5 &  -0.5 \\
		 -0.5 & 1 &  -0.5 \\
		  -0.5 &  -0.5 & 1
	\end{pmatrix}.
\end{align*}
The matrix ${R:=[\log(\vert P_{ij} \vert)]_{ij}}$ in this case is an EDM, but one can use the algorithm based on Theorem \ref{thm:GcharNew} to show that $P \notin G(\mathscr{C}_n^n)$. However, the converse does hold when all the elements of the Gram matrix $P$ are real and positive. 
\begin{corollary}
	If $P \in \text{Herm}(\mathbb{C}^n)$ and $P_{ij} \in \mathbb{R}$ and $P_{ij} >0$ for every $i,j \in [n]$, then $P \in G(\mathscr{C}_n^n)$ if an only if $\log \odot P$ is a Euclidean distance matrix ($\log$ function considered here is the real logarithm function). 
\end{corollary}
\begin{proof}
	The fact that if $P \in G(\mathscr{C}_n^n)$ then $\log \odot P$ is a Euclidean distance matrix can be seen using Corollary \ref{corr:NewCorr1} in this case. For the other direction, since $\log \odot P$ is an EDM, we can choose vectors $\lbrace \alpha_i \rbrace_{i=1}^n \subset \mathbb{R}^n$ such that for every $i,j \in [n]$
	\begin{align*}
		& \log P_{ij} = -\frac{1}{2}\Vert \alpha_i - \alpha_j\Vert_2^2 \\
		 \Leftrightarrow \ & P_{ij} =\exp \left(  -\frac{1}{2}\Vert \alpha_i - \alpha_j\Vert_2^2 \right),
	\end{align*}
	which implies that $P = G(\mket{\alpha_1},\ \cdots \ , \mket{\alpha_n})$ using the fact that $\alpha_i $ are real vectors and Eq. \ref{eq:ExpMaster2}.
\end{proof}

We can also connect our work to the coherent state mapping presented in Ref. \cite{Arrazola142}. We will assume that the entries of the Gram matrix, $P$ are close to $1$ (and $P_{ii}=1$ for all $ i \in [n]$). That is, the angles between the vectors forming the Gram matrix are small. In order to recreate this Gram matrix using multi-mode coherent states, we need to find amplitude vectors, $\lbrace \alpha_i \rbrace_{i=1}^n \subset \mathbb{C}^n$, a phase vector, $\phi \in \mathbb{R}^n$ and an integer matrix such that Eq. \ref{eq:GcharNew2} is satisfied. We choose the $[-\pi, \pi)$ branch of the $\log$ function. Further, for convenience we assume that no element of $P$ lies on the negative real axis. In this case, $\log \odot P$ is Hermitian, with zero diagonal. We will choose, $\lbrace \alpha_i \rbrace_{i=1}^n \subset \mathbb{C}^n$, such that $\Vert \alpha_i \Vert_2=1$ for all $i \in [n] $, $\phi = 0$, and $N=0$. For these choice, we need $\lbrace \alpha_i \rbrace_{i=1}^n$ satisfying
\begin{align*}
	\log (P_{ij}) & = - \frac{1}{2} \left( \Vert \alpha_i \Vert_2^2 + \Vert \alpha_j \Vert_2^2 - 2 \langle \alpha_i, \alpha_j \rangle\right)+  i \left( \phi_j -\phi_i \right)	\\
	& = - \frac{1}{2} \left( 2 - 2 \langle \alpha_i, \alpha_j \rangle\right).
\end{align*}
Using the assumption that the entries of $P$ are close to $1$, 
\begin{align*}
	& P_{ij}-1  + O ((P_{ij}-1)^2) = -1 + \langle \alpha_i, \alpha_j \rangle \\
	\Rightarrow &\ P_{ij} \approx \langle \alpha_i, \alpha_j \rangle.
\end{align*}
Thus, the problem of finding coherent states in this case reduces to finding unit vectors forming the Gram matrix, $P$. This can be done easily by choosing $\lbrace \alpha_i \rbrace_{i=1}^n$ to be the columns of $B$ for any $B$ such that $P= B^\dagger B$. Ref. \cite{Arrazola142} studies exactly this mapping of qudit states ($\lbrace \alpha_i \rbrace_i$) to multi-mode coherent states ($\lbrace \mket{\alpha_i} \rbrace_i$). Our results show that this is indeed well motivated. 
\end{sloppypar}

\section{Closure of the set of Gram matrices of multi-mode coherent states} \label{sec:Gclosure}

The set of Gram matrices of $n$ multi-mode coherent states, $G(\mathscr{C}_n^n)$ is not closed. For example, one may construct Gram matrices arbitrarily close to the identity matrix using coherent states, but the identity matrix itself cannot be constructed, since the inner products between any two coherent states is never zero. In this section, we will characterize the closure of $G(\mathscr{C}_n^n)$, which we will represent as $G\overline{(\mathscr{C}_n^n)}$. This set consists of the Gram matrices which can be approximated arbitrarily well using Gram matrices of coherent states. Experimentally $G\overline{(\mathscr{C}_n^n)}$ is more relevant than $G(\mathscr{C}_n^n)$. One expects that block diagonal Gram matrices, where each of the blocks is a Gram matrix of coherent states, would lie in $G\overline{(\mathscr{C}_n^n)}$. Each block could be realized by the set of corresponding coherent states, and one could displace the amplitudes between the sets relative to each other with a sufficiently large amplitude vector to achieve this. In fact, we will show that all the matrices in $G\overline{(\mathscr{C}_n^n)}$ can be put into such a block diagonal form. To prove this, we will require two intermediate results, Lemmas \ref{lemma:bounds} and \ref{lemma:NonZeroP}. Lemma \ref{lemma:bounds} relates the distance between the amplitude vectors of two coherent states with their inner-product with each other and a third coherent state. Lemma \ref{lemma:NonZeroP} shows that if a Gram matrix with non-zero entries belongs in $G\overline{(\mathscr{C}_n^n)}$, then it also belongs in $G{(\mathscr{C}_n^n)}$. Together these two will allow us to characterize $G\overline{(\mathscr{C}_n^n)}$.

\begin{sloppypar}
\begin{lemma}
	If $\mket{\alpha}, \mket{\beta}, \mket{\gamma} \in \mathscr{C}_n$, are such that $\left\vert \left\langle \mket{\alpha}, \mket{\beta} \right\rangle \right\vert = p_{\alpha \beta} $ and $\left\vert \left\langle \mket{\alpha}, \mket{\gamma} \right\rangle \right\vert  = p_{\alpha \gamma}$, then
	\begin{align}
		\left(-2 \log p_{\alpha \gamma} \right)^{1/2}&  + \left(-2 \log p_{\alpha \beta} \right)^{1/2}  \geq \Vert \beta - \gamma \Vert_2 \nonumber \\
		& \geq  \left\vert \left(-2 \log p_{\alpha \gamma} \right)^{1/2} - \left(-2 \log p_{\alpha \beta} \right)^{1/2} \right\vert.
	\end{align}
	\label{lemma:bounds}
\end{lemma}
\begin{proof}
	From Eq. \ref{eq:ExpMaster2}, the following can be deduced:
	\begin{align*}
		\Vert \alpha - \beta \Vert_2 &= \left( -2 \log \left( p_{\alpha \beta} \right) \right)^{1/2}  \\
		\Vert \alpha - \gamma \Vert_2 &= \left( -2 \log \left( p_{\alpha \gamma} \right) \right)^{1/2}.
	\end{align*}
	To establish the lower bound, we use the triangle inequality in the following manner. 
	\begin{align*}
		\Vert \gamma - \beta \Vert_2 &= \Vert \left( \gamma- \alpha \right) - \left( \beta- \alpha \right)\Vert_2 \\
		&\geq \vert \Vert \gamma- \alpha\Vert_2 - \Vert \beta- \alpha \Vert_2 \vert \\
		&\geq \vert \left(-2 \log p_{\alpha \gamma} \right)^{1/2} - \left(-2 \log p_{\alpha \beta} \right)^{1/2} \vert
	\end{align*}
	We use the triangle inequality again to establish the upper bound.
	\begin{align*}
		\Vert \gamma - \beta \Vert_2 &\leq \Vert \gamma- \alpha \Vert_2  +  \Vert \beta- \alpha \Vert_2 \\
		&= \left(-2 \log p_{\alpha \gamma} \right)^{1/2} + \left(-2 \log p_{\alpha \beta} \right)^{1/2}
	\end{align*}
\end{proof}
Before we go on further, we would like to point out that if ${P \in G(\mathscr{C}_n^n)}$, then ${P_{ij} \neq 0}$ for all ${i,j \in [n]}$. This can be seen from Eq. \ref{eq:ExpMaster1}. In the lemma that follows, we prove that if a matrix with non-zero entries belongs in $G\overline{(\mathscr{C}_n^n)}$, then it also belongs in $G{(\mathscr{C}_n^n)}$.

\begin{lemma}
	For $P \in \text{L}(\mathbb{C}^n)$ such that $ P_{ii}=1$ for all $i \in [n]$, if $ P_{ij}\neq 0$ for all ${i,j \in [n]}$, then $P \in \overline{G(\mathscr{C}_n^n)}$ if and only if $P \in G(\mathscr{C}_n^n)$.
	\label{lemma:NonZeroP}
\end{lemma}

\begin{proof}
${P \in G(\mathscr{C}_n^n)}$ implies ${P \in \overline{G(\mathscr{C}_n^n)}}$ is trivial. For the other direction, we have that, if ${P \in \overline{G(\mathscr{C}_n^n)}}$, then there exists a sequence, ${\lbrace P_k \rbrace_{k=1}^{\infty} \subseteq G(\mathscr{C}_n^n)}$ such that ${\lim_k P_k = P}$ under the 2-norm (where the 2-norm, $\Vert \cdot \Vert_2$ is defined as ${\Vert X\Vert_2 = (\tr(X^\dagger X))^{1/2}}$; all norms are equivalent in a finite dimensional space \cite{Bollobas99}, so we can choose the 2-norm without loss of generality). Firstly, observe that ${P \in \text{Pos}(\mathbb{C}^n)}$ and that ${P_{ii}=1}$ for all ${i \in [n]}$, since these sets are closed. Moreover, this implies the existence of integer matrices, ${\lbrace N_k \rbrace_{k=1}^\infty \subset \mathbb{Z}^{n \times n}}$, for which each $(\log \odot (P_k) -2 \pi i N_k) \in \text{Herm}(\mathbb{C}^n)$, ${2 \pi i (N_k)_{ii} = \log (1)}$ for every $i$ and such that for all $y \in \mathbb{C}^n$, $\langle u,y \rangle = 0$ (we are using an alternate characterization of Theorem \ref{thm:GcharNew}; see Lemma \ref{lemma:Qchar}, Statement $\rom{5}$)
	\begin{align*}
		y^\dagger (\log\odot(P_k)-2\pi iN_k) y \geq 0.
	\end{align*}
	Further, we can assume that for all $i,j$:
	\begin{align*}
		(N_k)_{ij} \in \mathbb{Z} \cap \left[ \frac{\beta}{2\pi} - \frac{1}{\pi}\vert \log (\delta_k) \vert - 1, \frac{\beta}{2\pi} + \frac{1}{\pi}\vert \log (\delta_k) \vert +2 \right)
	\end{align*}
	where $\delta_k : = \min_{i,j} |(P_k)_{ij}|$, and the branch of the function $\log$ is chosen such that it is continuous at $P_{ij}$ for every $i,j \in [n]$. Before we proceed further, we show that for sufficiently large $k$ we can restrict the choice of matrices, $N_k$ to a finite set. \\
	\begin{claim}
		For the sequence, $\lbrace P_k \rbrace_{k=1}^\infty $ and integer matrices $\lbrace N_k \rbrace_{k=1}^\infty$ as described above, there exists $M \in \mathbb{N}$, such that for all $k \geq M$ and for all $i,j \in [n]$, we have
		\begin{align*}
			(N_k)_{ij} \in \mathbb{Z} \cap \left[ \frac{\beta}{2\pi} - \frac{1}{\pi}\vert \log (\frac{\delta}{2}) \vert - 1, \frac{\beta}{2\pi} + \frac{1}{\pi}\vert \log (\frac{\delta}{2}) \vert +2 \right),
		\end{align*}
		where $\delta :=  \min_{i,j} |P_{ij}|$. This essentially says that the range of elements of $N_k$ can be made independent of $k$. 
	\end{claim}
	\begin{claimproof}
		Since, $\delta \neq 0$ and the elements ${(P_k)_{ij} \rightarrow P_{ij}}$ for every $i,j \in [n]$, we can choose $M$ such that for all $k \geq M$ and $i,j \in [n]$
		\begin{align*}
			& \vert \ \vert (P_k)_{ij} \vert - \vert P_{ij} \vert  \ \vert < \frac{\delta}{2} \\
			& \Rightarrow\ \vert P_{ij} \vert - \frac{\delta}{2} < \vert (P_k)_{ij} \vert \\
			& \Rightarrow\ \delta - \frac{\delta}{2} < \vert (P_k)_{ij} \vert \\
			& \Rightarrow\ \frac{\delta}{2} < \vert (P_k)_{ij} \vert .
		\end{align*}
		Since, this is true for all $i,j \in [n]$, for every $k \geq M$ we have 
		\begin{align*}
			\frac{\delta}{2} < \min_{i,j} \vert (P_k)_{ij} \vert = \delta_k.
		\end{align*}
		Using the fact that $\vert \log(x)  \vert $ is decreasing for $x \in (0,1]$, we see that
				\begin{align*}
			\left[ \frac{\beta}{2\pi} - \frac{1}{\pi}\vert \log (\delta_k) \vert - 1, \frac{\beta}{2\pi} + \frac{1}{\pi}\vert \log (\delta_k) \vert +2 \right)
		\end{align*}
		is a subset of 
		\begin{align*}
			\left[ \frac{\beta}{2\pi} - \frac{1}{\pi}\vert \log (\frac{\delta}{2}) \vert - 1, \frac{\beta}{2\pi} + \frac{1}{\pi}\vert \log (\frac{\delta}{2}) \vert +2 \right).
		\end{align*}
		This proves the claim for $M$ as chosen above.
	\end{claimproof}
	\\
	Now, we only consider the infinite sequence of matrices, ${\lbrace P_k : k \geq M \rbrace}$, where $M$ is the number as defined in the claim above. Observe that the integer matrix corresponding to each of these matrices, $N_k$, is chosen from the finite set of matrices, 
	\begin{align*}
		\mathcal{F} := & \lbrace N \in \mathbb{Z}^{n \times n} | \text{ for all } i,j \in [n]: \\
		 & \frac{\beta}{2\pi} - \frac{1}{\pi}\vert \log (\frac{\delta}{2}) \vert - 1 \leq N_{ij} < \frac{\beta}{2\pi} + \frac{1}{\pi}\vert \log (\frac{\delta}{2}) \vert +2 \rbrace.
	\end{align*}
	We can write 
	\begin{align*}
		\mathcal{F} = \lbrace K_1, K_2,\ \cdots\ , K_l \rbrace
	\end{align*}
	for some $l \in \mathbb{N}$ to emphasize the fact that its finite. Since, this set is finite and the sequence $(N_k)_k$ is infinite, there exists a $p \in [l]$ such that $N_k = K_p$ for infinitely many $k$. Define $K:=K_p$ for convenience. We choose a subsequence of $\lbrace P_k : k \geq M \rbrace$, $\lbrace P_{k_t}\rbrace_{t=1}^\infty$ such that $N_{k_t} = K $ for all $ t \geq 1$. Also, as this is a subsequence of $\lbrace P_k \rbrace_k$, $P_{k_t} \rightarrow P$. Now, observe that for every $y \in \mathbb{C}^n$ such that $\langle u, y \rangle = 0$ and for all $ t \geq 1$,
	\begin{align*}
		y^\dagger (\log \odot (P_{k_t}) - 2 \pi i K) y \geq 0.
	\end{align*}
	This implies that for such a vector, $y$, 
	\begin{align*}
		& \lim_{t \rightarrow \infty} y^\dagger (\log \odot (P_{k_t}) - 2 \pi i K) y \geq 0 \\
		& \Rightarrow\ y^\dagger (\log \odot (\lim_{t \rightarrow \infty}(P_{k_t})) - 2 \pi i K) y \geq 0 \\
		& \Rightarrow\ y^\dagger (\log \odot (P) - 2 \pi i K) y \geq 0, 
	\end{align*}
	where we have used the fact that $K$ is a constant and the functions-- ${f_y (X) = y^\dagger X y}$, and ${f(X) = \log \odot X}$ are continuous. The $\log \odot X$ function is continuous because we chose the branch of the $\log$ function such that no element of $P$ lay on the branch cut. Further, $2 \pi i (K)_{ii} = \log (1)$ for every $i$, since $K = N_k$ for some $k$, and  $\log \odot (P) - 2 \pi i K = \lim_{t \rightarrow \infty} (\log \odot (P_{k_t}) - 2 \pi i K) \in \text{Herm} (\mathbb{C}^n)$, because the set of Hermitian matrices is closed. Using an equivalent characterization of Theorem \ref{thm:GcharNew}, this proves our claim.
\end{proof}

\noindent We need one final notion to characterize $\overline{G(\mathscr{C}_n^n)}$. Observe that if 
\begin{align*}
	& G(v_1, \cdots, v_n) \in G(\mathscr{C}_n^n) \text{, then} \\
	& G(v_{\pi^{-1}(1)}, \cdots, v_{\pi^{-1}(n)}) = P_\pi G(v_1, \cdots, v_n) P_\pi^\dagger \in G(\mathscr{C}_n^n), \numberthis
	\label{eq:rearrangement}
\end{align*}
for any permutation $\pi$ (where $P_\pi$ represents the permutation matrix associated with $\pi$). If one can construct a Gram matrix, $Q$ using coherent states, then all one needs to do to construct the Gram matrix, $P_\pi Q P_\pi^\dagger$ is to permute the order of the coherent states forming $Q$ by $\pi$. Therefore, ${Q \in G(\mathscr{C}_n^n)}$ is equivalent to ${P_\pi Q P_\pi^\dagger \in G(\mathscr{C}_n^n)}$. In fact, because of the isometric invariance \cite{Watrous18} of the 2-norm, ${Q \in \overline{G(\mathscr{C}_n^n)}}$ is equivalent to ${P_\pi Q P_\pi^\dagger \in \overline{G(\mathscr{C}_n^n)}}$. We can use this fact to simplify our analysis of matrices in $\overline{G(\mathscr{C}_n^n)}$. In the rest of the paper, we will refer to a Gram matrix of the form, $Q^\prime = G(v_{\pi^{-1}(1)}, v_{\pi^{-1}(2)}, \cdots, v_{\pi^{-1}(n)})$ as a rearrangement of the vectors forming the Gram matrix, $Q= G(v_1, v_2, \cdots, v_n)$.

\begin{thm} A matrix, $P \in \text{L}(\mathbb{C}^n)$, is a member of $\overline{G(\mathscr{C}_n^n)}$ if and only if $P$ can be written as 
	\begin{align}
		P_{\pi} P P_{\pi}^\dagger= \bigoplus_{i=1}^m P_i = 
		\begin{pmatrix}
			P_1 & & & \\
			& P_2 & & \\
			& & \ddots & \\
			& & & P_m
		\end{pmatrix}
		\label{eq:closureG}
	\end{align}
	where, ${\bigoplus_{i=1}^m P_i}$ represents a direct sum of matrices, ${\lbrace P_i \rbrace_{i=1}^m}$ and for each ${i \in [m]}$, ${P_i \in G(\mathscr{C}_{n_i}^{n_i})}$ for some $n_i \in \mathbb{N}$, and $P_\pi$ is a permutation matrix. \\
	
	\noindent In other words, $P \in \overline{G(\mathscr{C}_n^n)}$ if and only if up to a rearrangement of the vectors forming it, P can be written as a block-diagonal matrix where each block is a Gram matrix that can be realized by multi-mode coherent states.
	\label{thm:closureG}
\end{thm}
\begin{proof}
	
	We will first prove that if $P \in \overline{G(\mathscr{C}_n^n)}$, then it has the aforementioned block diagonal form. This will be done in two steps. In the first step, we will establish two properties of elements of such a matrix, $P$. In the second step, which primarily relies on linear algebra, we will show that these two properties suffice to prove that the matrix, $P$, has the required block diagonal structure. \\
	
	\emph{Step 1:} For a ${P \in \overline{G(\mathscr{C}_n^n)}}$, and indices $i,k \in [n]$ such that ${P_{ik} \neq 0}$ we will prove that, if ${ j \in [n]}$ such that ${ P_{ij}=0}$, then ${P_{kj}=0}$ and if ${j \in [n]}$ such that ${ P_{ij} \neq 0}$, then ${ P_{kj}\neq 0}$. The idea is that if two multi-mode coherent states have a non-zero inner product then their amplitude vectors have to be a finite distance away from each other, but if their inner product approaches zero then the distance between these vectors has to grow infinitely large. \\
	
	\noindent We will first show that if $P_{ik} \neq 0$ and $j \in [n]$ such that ${ P_{ij}=0}$, then ${P_{kj}=P_{jk}=0}$. To prove this pick a sequence, ${\lbrace P_u \rbrace_{u=1}^{\infty} \subseteq G(\mathscr{C}_{n}^{n})}$ such that,
	\begin{align*}
		\lim_{u \rightarrow \infty} P_u = P.
	\end{align*}
	Here, we once again consider the 2-norm without any loss of generality. Observe that this is equivalent to 
	\begin{align*}
		\lim_{u \rightarrow \infty} \left( P_u \right)_{ab} = P_{ab}
	\end{align*}
	for all $a,b \in [n]$. Since, ${P_u \in G(\mathscr{C}_{n}^{n})}$, there exist multi-mode coherent states, ${\left\lbrace e^{i \phi_{ui}} \mket{\alpha^u_i} \right\rbrace_{i=1}^{n} \subseteq \mathscr{C}_{n}}$, such that ${\left( P_u \right)_{ab} = \left\langle e^{i \phi_{ua}} \mket{\alpha_a^u}, e^{i \phi_{ub}} \mket{\alpha_b^u} \right\rangle}$ for all $ a,b \in [n]$. Using Lemma \ref{lemma:bounds}, we have,
	\begin{align*}
		& \left( - 2 \log \vert \left( P_u \right)_{jk} \vert \right)^{1/2} = \Vert \alpha^u_{k} - \alpha^u_j \Vert_2 \\ 
		& \qquad \geq \left(-2 \log \vert \left( P_u \right)_{ij} \vert \right)^{1/2} - \left(-2 \log \vert \left( P_u \right)_{ik} \vert \right)^{1/2}.
		 \end{align*}
		 Taking the limit, $u \rightarrow \infty$, we have
		 \begin{align*}
		 \lim_{u \rightarrow \infty} \left( - 2 \log \vert \left( P_u \right)_{jk} \vert \right)^{1/2} & \geq \left(-2 \log \lim_{u \rightarrow \infty} \vert \left( P_u \right)_{ij} \vert \right)^{1/2} \\
		 & \qquad - \left(-2 \log \lim_{u \rightarrow \infty} \vert \left( P_u \right)_{ik} \vert \right)^{1/2} \\
		 & \geq \left(-2 \log \lim_{u \rightarrow \infty} \vert \left( P_u \right)_{ij} \vert \right)^{1/2} \\
		 & \qquad - \left(-2 \log \vert P_{ik} \vert \right)^{1/2},
	\end{align*}
	where we have used the continuity of $\vert \cdot \vert$, $\log(\cdot)$, and $(\cdot)^{1/2}$ functions. The limit on the RHS tends to $\infty$ since ${\lim_{u \rightarrow \infty} \vert \left( P_u \right)_{ij} \vert = \vert P_{ij} \vert= 0}$ and $P_{ik} \neq 0$, therefore
	\begin{align*}
			& \left( - 2 \log \lim_{u \rightarrow \infty} \vert \left( P_u \right)_{jk} \vert \right)^{1/2} = \infty \\
			& \Rightarrow \lim_{u \rightarrow \infty} \vert \left( P_u \right)_{jk} \vert = \vert P_{jk} \vert = 0.
	\end{align*}	 
	Hence, for $P_{ik} \neq 0$, and for every $ j \in [n]$ such that $P_{ij}=0$,
	\begin{align}
			 P_{kj}=P_{jk}=0. \label{eq:PijZero}
	\end{align}
	since, $P$ is Hermitian. \\
	
	\noindent Now, we will prove that given $P_{ik} \neq 0$, if for ${ j \in [n]}$ ${ P_{ij} \neq 0}$, then ${P_{kj}\neq 0}$. For $j$ such that $P_{ij} \neq 0$, using the upper bound given in Lemma \ref{lemma:bounds}, we have 
	\begin{align*}
		& \left( - 2 \log \vert \left( P_u \right)_{jk} \vert \right)^{1/2} = \Vert \alpha_{k}^u -\alpha_j^u \Vert_2 \\
		& \qquad \leq \left( -2 \log \left(\left\vert (P_u)_{ki} \right\vert  \right) \right)^{1/2} + \left( -2 \log \left(\left\vert (P_u)_{ij} \right\vert  \right) \right)^{1/2}.
	\end{align*}
	Taking the limit, $u \rightarrow \infty$, we have
	\begin{align*}
		\lim_{u \rightarrow \infty} \left( - 2 \log \vert \left( P_u \right)_{jk} \vert \right)^{1/2} & \leq \lim_{u \rightarrow \infty} \left( -2 \log \left(\vert (P_u)_{ki} \vert  \right) \right)^{1/2} \\
		& \qquad + \lim_{u \rightarrow \infty} \left( -2 \log \left(\vert (P_u)_{ij} \vert  \right) \right)^{1/2} \\
		& \leq \left( -2 \log \vert P_{ki} \vert \right)^{1/2} \\
		& \qquad + \left( -2 \log \vert P_{ij} \vert  \right)^{1/2} =: M < \infty 
	\end{align*}
	where we have defined $M$ to be the upper bound. This gives us
	\begin{align*}
		\vert P_{jk} \vert & \geq \exp \left( - \frac{1}{2} M^2 \right) > 0.
	\end{align*}
	Therefore, for ${ P_{ik} \neq 0}$ and ${j \in [n]}$ such that ${P_{ij} \neq 0}$, we have
	\begin{align}
		 P_{jk} = P_{kj}^\ast \neq 0.
		\numberthis
		\label{eq:PijNotZero}
	\end{align}	
	Therefore, we have proven that for any $n \in \mathbb{N}$, a matrix, ${P \in \overline{G(\mathscr{C}_n^n)}}$, and $i,j,$ and $k \in [n]$ Eq. \ref{eq:PijZero} and Eq. \ref{eq:PijNotZero} hold. These two properties are sufficient to establish the block diagonal structure of $P$. \\
	
	\emph{Step 2:} We will use induction on the size of the Gram matrices to prove the statement that if ${P \in \overline{G(\mathscr{C}_n^n)}}$, then $P$ has the block diagonal form given in Eq. \ref{eq:closureG}, up to a rearrangement of the vectors forming it. First observe that since the sets, ${\lbrace X \in \text{L}(\mathbb{C}^n) : X_{ii}=1 \text{ for all } i \in [n] \rbrace}$ and $\text{Pos}(\mathbb{C}^n)$ are closed, $P$ will belong in these sets. The induction hypothesis is clearly true for $n=1$ as $P=\left( 1 \right)$ is the only Gram matrix in this case and ${P = G(\mket{0})}$. We assume that our hypothesis is true for all $p \leq n$. For ${P \in \overline{G(\mathscr{C}_{n+1}^{n+1})}}$, $P$ can always be put into the form
	\begin{align*}
	P = 
		\begin{pmatrix}
		P^\prime & x \\
		x^\dagger & 1
		\end{pmatrix}
	\end{align*}
	where $x \in \mathbb{C}^n$. It can be shown that ${P^\prime \in \overline{G(\mathscr{C}_n^n)}}$. Since, ${P}$ is positive semidefinite, we can write ${P = G(v_1, v_2, \cdots, v_{n+1})}$ for vectors, ${\lbrace v_i \rbrace_{i=1}^{n+1}}$. Then, ${P^\prime = G(v_1, v_2, \cdots, v_{n})}$. By the induction hypothesis, there exists a permutation, $\pi$ such that 
	\begin{align}
	P_{\pi} P^\prime P_{\pi}^\dagger = G (v_{\pi^{-1}(1)}, \cdots, v_{\pi^{-1}(n)}) =
		\begin{pmatrix}
			P_1^\prime & & & \\
			& P_2^\prime & & \\
			& & \ddots & \\
			& & & P_m^\prime
		\end{pmatrix} 
		\label{eq:HypthBlockMatrix1}
	\end{align}
	where for every ${i \in [m]}$, ${P^\prime_i \in G(\mathscr{C}_{n_i}^{n_i})}$ for some $n_i$. We can transform $P$ as
	\begin{align*}
		P \longrightarrow 
		\begin{pmatrix}
			P_{\pi} & 0 \\
			0 & 1
		\end{pmatrix} P 
		\begin{pmatrix}
			P^\dagger_{\pi} & 0 \\
			0 & 1
		\end{pmatrix}
	\end{align*}
	and prove our claim for this matrix without loss of generality. So, from now on we will assume that $P$ is block diagonal in the first $n \times n $ entries. \\
	
	If for every ${ i \in [n]}$, ${P_{(n+	1)i} =0}$, then our matrix is already in the required block diagonal form. So, we will assume that there exists ${i \in [n]}$ such that $P_{(n+1)i} \neq 0$. Observe that the block structure of the first $n \times n$ entries divides the vectors forming the Gram matrix into orthogonal subspaces. So, we may associate a subspace with each Gram matrix, $P^\prime_l$ (Eq. \ref{eq:HypthBlockMatrix1}). Further assume without loss of generality that the vector, $v_i$ (where $P= G(v_1, v_2, \cdots, v_{n+1})$) is in the subspace associated with the Gram matrix, $P^\prime_m$ (if not one can always permute the vectors forming $P$ such that this is true). Then using the fact that for all $j$ such that ${P_{ij} = 0}$, ${P_{(n+1)j} =0}$ (Eq. \ref{eq:PijZero}), one can see that $P$ also has a block diagonal form if one includes the $(n+1)^\text{th}$ row and column in $P^\prime_m$ (See Fig. \ref{fig:closureProof} for a schematic representation of this fact). We will call this new last block $P_m$. All the other blocks remain the same. \\
	
	\begin{figure}[htb]
	\centering
		\includegraphics[scale=0.3]{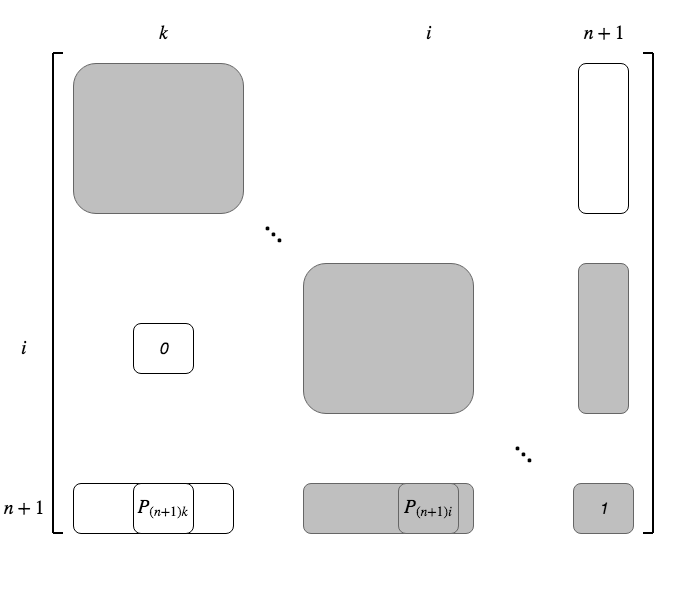}
		\caption{In this figure, we schematically represent the matrix, $P \in \overline{G(\mathscr{C}_{n+1}^{n+1})}$, which is diagonal in its first $n \times n$ entries. We consider the case where ${P_{(n+1)i} \neq 0}$. The facts proved in Step 1 of the proof show that the zero (white) and non-zero (gray) terms in the $i^\text{th}$-row (column) coincide with zero and non-zero terms respectively in the $(n+1)^\text{th}$-row (column). }		
		\label{fig:closureProof}
	\end{figure}		
	
	Furthermore, for every ${l \in [m-1]}$, ${P^\prime_l \in G(\mathscr{C}_{n_l}^{n_l})}$, and if we prove that the new block, ${P_m}$ belongs in ${ G(\mathscr{C}_{n_m+1}^{n_m+1})}$ then we would be done. Observe that ${P^\prime_m \in G(\mathscr{C}_{n_m}^{n_m})}$, which means that ${\left( P^\prime_m \right)_{uv} \neq 0}$ for all $u,v \in [n_m]$. In addition, using Eq. \ref{eq:PijNotZero}, we can show that for all ${u,v \in [n_m+1]}$, ${\left(P_{m} \right)_{uv} \neq 0}$. Moreover, ${P_m \in \overline{G(\mathscr{C}_{n_m+1}^{n_m+1})}}$. Using Lemma \ref{lemma:NonZeroP}, these two imply that ${P_m \in G(\mathscr{C}_{n_m+1}^{n_m+1})}$. Therefore, if $P \in \overline{G(\mathscr{C}_n^n)}$, then $P$ has the block diagonal form given in Eq. \ref{eq:closureG}, up to a rearrangement of the vectors forming it.\\
	
	We will prove the converse of the statement by construction. We will present a construction for Gram matrices with 2 blocks, which can be generalised to $m$ blocks easily. Suppose, we have ${P_1 \in G(\mathscr{C}_{n_1}^{n_1})}$ and ${P_2 \in G(\mathscr{C}_{n_2}^{n_2})}$, then we wish to prove that
	\begin{align*}
		P =
		\begin{pmatrix}
			P_1 & \\
			 &  P_2
		\end{pmatrix} \in G(\mathscr{C}_{n}^{n})
	\end{align*}
	for $n = n_1 + n_2$. \\
	 \noindent For $i=1,2$ let the multi-mode coherent states, ${\left\lbrace e^{i \phi_{ij}} \mket{\alpha^i_j} \right\rbrace_{j=1}^{n_i}}$ be such that
	\begin{align*}
		\left( P_i\right)_{uv}  = \left\langle e^{i \phi_{iu}} \mket{\alpha^i_u}, e^{i \phi_{iv}} \mket{\alpha^i_v}\right\rangle.
	\end{align*}
	There are at least two ways in which this can be accomplished. One way would be to put both the sets of amplitude vectors into the same space, say ${\mathcal{X} = \mathbb{C}^{n^\prime}} $ for ${{n^\prime = \max\lbrace n_1, n_2\rbrace}}$, and to displace one of the sets by a large amplitude vector, $A \in \mathcal{X}$. This way the inner products of the coherent states belonging to the same set of states remains invariant for appropriately defined phases, but the inner product of the coherent states belonging to different sets would tend to zero as the norm of the vector, $A$, tends to $\infty$. The second way, which we present here, is similar but it puts the amplitude vectors in different subspaces, and makes the distance between these subspaces go to $\infty$.\\
	
	We will define multi-mode coherent states $\left\lbrace e^{i \Phi_j}\mket{\beta_j (A)} \right\rbrace_{j=1}^{n}$ dependent on a parameter, $A \in \mathbb{R}$, such that their Gram matrix will approach $P$ as $A \rightarrow \infty$. Define,
	\begin{align*}
		& n := n_1 + n_2 
	\end{align*}
	and,
	\begin{align*}
		& \left\lbrace \beta_j (A) \right\rbrace_{j=1}^{n} \subset \left( \mathbb{C}^{n_1} \oplus \mathbb{C}^{n_2} \right) \oplus \left( \mathbb{C} \oplus \mathbb{C} \right) \\
		& \left\lbrace \Phi_j \right\rbrace_{j=1}^{n} \subset \mathbb{R}
	\end{align*}
	such that for each $j \in [n]$,
	\begin{align*}
		& \beta_j (A):= 
			\begin{cases}
				\alpha_j^1 \oplus 0 \oplus A \oplus 0 & j \leq n_1 \\
				0 \oplus \alpha_{j^\prime}^2\oplus 0 \oplus A & j^\prime = j- n_1 > 0 \\
			\end{cases} \\
		& \Phi_j := 
			\begin{cases}
				\phi_{1j} & j \leq n_1 \\
				\phi_{2j^\prime} & j^\prime = j- n_1 > 0
			\end{cases}
	\end{align*}
	Given these, one can check that the following hold for $u, v \in [n_1]$ 
	\begin{align*}
		& \Vert \beta_u (A) - \beta_v (A) \Vert^2_2 =  \Vert \alpha^1_{u} - \alpha^1_{v} \Vert^2_2 \\
		& \text{Im}\left\lbrace \langle \beta_u (A), \beta_v (A) \rangle \right\rbrace = \text{Im}\left\lbrace \langle \alpha^1_{u}, \alpha^1_{v} \rangle \right\rbrace \\
		& \Phi_u -\Phi_v = \phi_{1u} - \phi_{1v},
	\end{align*}
	whic imply that,
	\begin{align*}
		& \langle e^{i \Phi_u} \mket{\beta_u (A)}, e^{i \Phi_v } \mket{\beta_v (A)} \rangle = \langle e^{i \phi_{1{u}}} \mket{\alpha^1_{u}}, e^{i \phi_{1v}} \mket{\alpha^1_{v}} \rangle.
	\end{align*}
	Similar relations hold for the case when $u, v > n_1$, although one needs to replace $u$ with $u^\prime = u - n_1$ and $v$ with $v^\prime = v- n_1$ on the RHS of these equations. For $u \leq n_1$ and $v > n_1$, we have 
	\begin{align*}
	 	& \Vert \beta_u (A) - \beta_v (A) \Vert^2_2 = \Vert \alpha^1_{u}\Vert^2_2 + \Vert \alpha^2_{v^\prime} \Vert^2_2 +2A^2 \geq 2A^2 \\
	 	& \Rightarrow \vert \langle e^{i \Phi_u} \mket{\beta_u (A)}, e^{i \Phi_v } \mket{\beta_v (A)} \rangle \vert \leq \exp (-A^2) . \numberthis \label{eq:expConv}
	\end{align*}
	The matrix, ${P(A) = G(e^{i \Phi_1} \mket{\beta_1 (A)}, \cdots, e^{i \Phi_n} \mket{\beta_n (A)})}$ is a member of $ {G\left( \mathscr{C}_{n+2}^n \right)}$ and also of ${G\left( \mathscr{C}_{n}^n \right)}$ (by Corollary \ref{corr:GOnlyN}) for every $A \in \mathbb{R}$. For this family of matrices, we have that
	\begin{align*}
		\lim_{A \rightarrow \infty} P(A) = 
		\begin{pmatrix}
			P_1 & \\
			& P_2
		\end{pmatrix} \in \overline{G\left( \mathscr{C}_{n}^n \right)}.
	\end{align*}
	This together with the observation that ${P \in \overline{G(\mathscr{C}_n^n)} \Leftrightarrow P_\pi P P_\pi^\dagger \in \overline{G(\mathscr{C}_n^n)}}$ for a permutation matrix, $P_\pi$, completes our proof. 
	
\end{proof}
We have proven here that an $n \times n$ Gram matrix can be approximated arbitrarily well using Gram matrices of multi-mode coherent states if and only it can be put into the form of Eq. \ref{eq:closureG}. Moreover, we have also shown that if this is the case then we can approximate it using at most $n-$mode coherent states. During the proof of the converse of the theorem, we also suggest a way to construct such matrices, which would potentially require only $\max_i \lbrace{n_i}\rbrace$ number of modes (here $n_i$ are the dimensions of the block diagonal matrices in Eq. \ref{eq:closureG}). Secondly, one may wish to understand the energy requirements for approximating a Gram matrix up to an error, $\epsilon$, using coherent states. In the proof for the converse presented here, the overlap between the states of two blocks decreases exponentially fast in the number of additional photons (Eq. \ref{eq:expConv}). Therefore, we would only require $O(\log(1/ \epsilon))$ additional photons to approximate a block diagonal Gram matrix.\\

We use this characterization of the closure of $G(\mathscr{C}_n^n)$ to show that one cannot arbitrarily approximate Gram matrices of mutually unbiased bases using multi-mode coherent states. Recall that two orthonormal sets of vectors, $\lbrace \ket{e_i} \rbrace_{i=1}^n $ and $\lbrace \ket{f_i} \rbrace_{i=1}^n $ are said to be mutally unbiased \cite{Bengtsson07} if for every $i, j \in [n]$
\begin{align}
	\vert \langle \ket{e_i}, \ket{f_j} \rangle \vert = \frac{1}{n}.
	\label{eq:MUBCond}
\end{align}
In standard literature, these sets are bases of $ \mathbb{C}^n$. However, we relax this condition here and consider these to be sets of vectors in the infinite dimensional Fock space. 
\begin{example} \emph{Gram matrices of mutually unbiased bases cannot be arbitrarily approximated using multi-mode coherent states} \\
	Consider two mutually unbiased sets of vectors, ${\mathcal{E} = \lbrace \ket{e_i} \rbrace_{i=1}^n}$ and ${\mathcal{F}=\lbrace \ket{f_i} \rbrace_{i=1}^n}$ in the Fock space. $\mathcal{E}$ and $\mathcal{F}$ are orthonormal sets satisfying Eq. \ref{eq:MUBCond}. For these vectors, define ${P:=G(\ket{e_1}, \ket{e_2}, \ \cdots\ , \ket{e_n}, \ket{f_1}, \ket{f_2},\ \cdots\ ,\ket{f_n})}$. Suppose, ${P \in \overline{G(\mathscr{C}_{2 n}^{ 2 n})}}$, then using the fact that $P_{ki} \neq 0$ and $P_{ji} \neq 0$ implies $P_{kj} \neq 0$ for such matrices (equivalent to Eq. \ref{eq:PijNotZero} in Theorem \ref{thm:closureG}), we have that since $P_{1(n+1)} = \langle \ket{e_1}, \ket{f_1}\rangle \neq 0 $ and $P_{2(n+1)}  = \langle \ket{e_2}, \ket{f_1}\rangle \neq 0$, $P_{12} = \langle \ket{e_1}, \ket{e_2}\rangle  \neq 0 $, which is a contradiction. Hence, ${P \notin \overline{G(\mathscr{C}_{2 n}^{ 2 n})}}$.
\end{example}
\end{sloppypar}
\section{Conclusion}

In this paper, we have successfully characterized the set of Gram matrices of multi-mode coherent states and its closure. We provide tests to check if a Gram matrix belongs to either of these sets. We proved that no more than $(n-1)$-modes are required to represent a Gram matrix of $n$-vectors. These results will hopefully serve as a toolbox for formulating quantum protocols in terms of coherent states, and facilitate their experimental implementation. They also add to our theoretical knowledge of coherent states, and completely describe sets of states attainable from them. We also expect our results to be beneficial towards understanding the kind of quantum resources a communication protocol requires. 

\begin{acknowledgments}
We would like to thank Benjamin Lovitz and John Watrous for helpful discussions. We would also like to thank Benjamin Lovitz for pointing an error in an early version of the proof of Theorem \ref{thm:GcharNew}. We are also grateful to Matthias Kleinmann for independently pointing out an error in an earlier version of the paper. The work has been performed at the Institute for Quantum Computing, University of Waterloo, which is supported by Industry Canada. The research has been supported by NSERC under the Discovery Program, grant number 341495 and by the ARL CDQI program.
\end{acknowledgments}
\bibliographystyle{unsrt}
\bibliography{bib}

\begin{thebibliography}{10}

\bibitem{Buhrman98}
Harry Buhrman, Richard Cleve, and Avi Wigderson.
\newblock Quantum vs. classical communication and computation.
\newblock In {\em Proceedings of the Thirtieth Annual ACM Symposium on Theory
  of Computing}, STOC '98, pages 63--68, New York, NY, USA, 1998. ACM.

\bibitem{Buhrman01}
Harry Buhrman, Richard Cleve, John Watrous, and Ronald de~Wolf.
\newblock Quantum fingerprinting.
\newblock {\em Phys. Rev. Lett.}, 87:167902, Sep 2001.

\bibitem{Yossef04}
Ziv Bar-Yossef, T.~S. Jayram, and Iordanis Kerenidis.
\newblock Exponential separation of quantum and classical one-way communication
  complexity.
\newblock In {\em Proceedings of the Thirty-sixth Annual ACM Symposium on
  Theory of Computing}, STOC '04, pages 128--137, New York, NY, USA, 2004. ACM.

\bibitem{Gavinsky07}
Dmitry Gavinsky, Julia Kempe, Iordanis Kerenidis, Ran Raz, and Ronald de~Wolf.
\newblock Exponential separations for one-way quantum communication complexity,
  with applications to cryptography.
\newblock In {\em Proceedings of the Thirty-ninth Annual ACM Symposium on
  Theory of Computing}, STOC '07, pages 516--525, New York, NY, USA, 2007. ACM.

\bibitem{Grosshans02}
Fr{\'e}d{\'e}ric Grosshans and Philippe Grangier.
\newblock Continuous variable quantum cryptography using coherent states.
\newblock {\em Physical review letters}, 88:057902, 2002.

\bibitem{Grosshans03}
Fr{\'e}d{\'e}ric Grosshans, Gilles Van~Assche, J{\'e}r{\^o}me Wenger, Rosa
  Brouri, Nicolas~J. Cerf, and Philippe Grangier.
\newblock Quantum key distribution using gaussian-modulated coherent states.
\newblock {\em Nature}, 421, 01 2003.

\bibitem{Silberhorn02}
Ch. Silberhorn, T.~C. Ralph, N.~L\"utkenhaus, and G.~Leuchs.
\newblock Continuous variable quantum cryptography: Beating the 3 db loss
  limit.
\newblock {\em Phys. Rev. Lett.}, 89:167901, Sep 2002.

\bibitem{Arrazola14}
Juan~Miguel Arrazola and Norbert L\"utkenhaus.
\newblock Quantum fingerprinting with coherent states and a constant mean
  number of photons.
\newblock {\em Phys. Rev. A}, 89:062305, Jun 2014.

\bibitem{Xu15}
Feihu Xu, Juan~Miguel Arrazola, Kejin Wei, Wenyuan Wang, Pablo Palacios-Avila,
  Chen Feng, Shihan Sajeed, Norbert L{\"u}tkenhaus, and Hoi-Kwong Lo.
\newblock Experimental quantum fingerprinting with weak coherent pulses.
\newblock {\em Nature Communications}, 6:8735 EP --, 10 2015.

\bibitem{Guan16}
Jian-Yu Guan, Feihu Xu, Hua-Lei Yin, Yuan Li, Wei-Jun Zhang, Si-Jing Chen,
  Xiao-Yan Yang, Li~Li, Li-Xing You, Teng-Yun Chen, Zhen Wang, Qiang Zhang, and
  Jian-Wei Pan.
\newblock Observation of quantum fingerprinting beating the classical limit.
\newblock {\em Phys. Rev. Lett.}, 116:240502, Jun 2016.

\bibitem{Amiri17}
Ryan Amiri and Juan~Miguel Arrazola.
\newblock Quantum money with nearly optimal error tolerance.
\newblock {\em Phys. Rev. A}, 95:062334, Jun 2017.

\bibitem{Guan18}
Jian-Yu Guan, Juan~Miguel Arrazola, Ryan Amiri, Weijun Zhang, Hao Li, Lixing
  You, Zhen Wang, Qiang Zhang, and Jian-Wei Pan.
\newblock Experimental preparation and verification of quantum money.
\newblock {\em Phys. Rev. A}, 97:032338, Mar 2018.

\bibitem{Arrazola16}
Juan~Miguel Arrazola, Markos Karasamanis, and Norbert L\"utkenhaus.
\newblock Practical quantum retrieval games.
\newblock {\em Phys. Rev. A}, 93:062311, Jun 2016.

\bibitem{Gerry04}
Christopher Gerry and Peter Knight.
\newblock {\em Introductory Quantum Optics}.
\newblock Cambridge University Press, 2004.

\bibitem{Hillery85}
Mark Hillery.
\newblock Classical pure states are coherent states.
\newblock {\em Physics Letters A}, 111(8):409 -- 411, 1985.

\bibitem{Hudson74}
R.L. Hudson.
\newblock When is the wigner quasi-probability density non-negative?
\newblock {\em Reports on Mathematical Physics}, 6(2):249 -- 252, 1974.

\bibitem{Chefles04}
Anthony Chefles, Richard Jozsa, and Andreas Winter.
\newblock On the existence of physical transformations between sets of quatum
  states.
\newblock {\em International Journal of Quantum Information}, 02(01):11--21,
  2004.

\bibitem{Chefles00}
Anthony Chefles.
\newblock Deterministic quantum state transformations.
\newblock {\em Physics Letters A}, 270(1):14 -- 19, 2000.

\bibitem{Werner06}
Werner Vogel and Dirk-Gunnar Welsch.
\newblock {\em Quantum Optics}.
\newblock John Wiley \& Sons, 2006.

\bibitem{Gower82}
J.C. Gower.
\newblock Euclidean distance geometry.
\newblock {\em Math. Scientist}, 7:1--14, 1982.

\bibitem{Dattorro17}
Jon Dattorro.
\newblock {\em Convex Optimisation and Euclidean Distance Geometry}.
\newblock Meboo Publishing USA, 2017.

\bibitem{Arrazola142}
Juan~Miguel Arrazola and Norbert L\"utkenhaus.
\newblock Quantum communication with coherent states and linear optics.
\newblock {\em Phys. Rev. A}, 90:042335, Oct 2014.

\bibitem{Bollobas99}
B\'ela Bollob\'as.
\newblock {\em Linear Analysis: An Introductory Course}.
\newblock Cambridge University Press, 2 edition, 1999.

\bibitem{Watrous18}
John Watrous.
\newblock {\em The Theory of Quantum Information}.
\newblock Cambridge University Press, 2018.

\bibitem{Bengtsson07}
Ingemar Bengtsson.
\newblock Three ways to look at mutually unbiased bases.
\newblock {\em AIP Conference Proceedings}, 889(1):40--51, 2007.

\end{thebibliography}

\end{document}